\documentclass[letterpaper, 10 pt, conference]{ieeeconf}  

\IEEEoverridecommandlockouts                              

\usepackage{amsmath} 
\usepackage{amssymb}  
\usepackage{amsfonts}
\let\labelindent\relax
\usepackage[shortlabels]{enumitem}
\usepackage{graphicx}
\usepackage{multicol}
\usepackage{color}
\usepackage{cancel}
\usepackage{caption}
\usepackage{todonotes}
\usepackage[hidelinks]{hyperref}
\usepackage{booktabs}

\newtheorem{proposition}{Proposition}
\newtheorem{corollary}{Corollary}

\newtheorem{lemma}{Lemma}
\newtheorem{remark}{Remark}

\title{\LARGE \bf
Total energy-shaping control for mechanical systems via Control-by-Interconnection
}

\author{Joel Ferguson$^{1}$
\thanks{$^{1}$Joel Ferguson is with the School of Engineering, The University of Newcastle, Australia
        {\tt\small Email: joel.ferguson@newcastle.edu.au}}%
}

\begin{document}

\maketitle
\thispagestyle{empty}
\pagestyle{empty}

\begin{abstract}

Application of IDA-PBC to mechanical systems has received much attention in recent decades, but its application is still limited by the solvability of the so-called matching conditions. In this work, it is shown that total energy-shaping control of under-actuated mechanical systems has a control-by-interconnection interpretation. Using this interpretation, alternate matching conditions are formulated that defines constraints on the added energy, rather then the total closed-loop energy. It is additionally shown that, for systems that are under-actuated degree one with the mass matrix depending on a single coordinate, the kinetic energy matching conditions resolve to ODEs which can be evaluated numerically. Using this approach controllers are proposed for the benchmark cart-pole and acrobot systems.

\end{abstract}

\section{INTRODUCTION}
Energy-based methods for controlling nonlinear physical systems have been shown to be effective in a variety of physical domains \cite{Ortega2004}. Such methods consider the energy and structure of the system to be controlled to derive control strategies that exploit the natural system behaviours. Interconnection and Damping Assignment, Passivity-Based Control (IDA-PBC) is one such control methodology where the control input is designed such that the closed-loop can be interpreted as an alternate physical system with a different energy, interconnection and damping structure \cite{Ortega2002a}.

While IDA-PBC has been applied to a broad range of systems, particular attention has been given to mechanical systems which exhibit a rich canonical structure \cite{Gomez-Estern2001,Acosta2005,Mahindrakar2010}. In the case of fully-actuated systems, IDA-PBC allows a user to arbitrarily modify the potential and kinetic energy of the closed-loop system \cite{Arpenti2022}, a process known as total energy shaping \cite{Gomez-Estern2001}. For under-actuated systems, however, application of IDA-PBC is limited by solutions that satisfy a set of PDEs, the so-called matching conditions. Much research effort has been committed to solving these equations, with solutions posed in several special cases \cite{Acosta2005,Mahindrakar2010,Arpenti2022}. This design methodology has been applied to a number of benchmark examples such as the cart-pole, acrobot, spider crane, amongst others.

Control-by-Interconnection (CbI) describes a sub-class of energy-based control methods that falls under the umbrella of IDA-PBC \cite{Ortega2008,Ortega2014}. Under this scheme, the controller is assumed to be a passive system that is interconnected with a passive plant to be controlled via the passive input-output pair. Casimirs, conserved quantities between the control sub-system and the plant, can be constructed to help shape the energy of the closed-loop system. It is known that potential energy shaping of fully-actuated mechanical systems falls into the class of CbI \cite{Ortega2008,Duindam2009}. Control of underactuated mechanical systems has been explored in the context of CbI by applying nonlinear PID controllers to both the standard and alternate passive outputs \cite{Donaire2016c,Romero2016,Romero2017a}. The idea of using PID for stabilisation of passive systems was formalised in \cite{Zhang2017} and a general characterisation of all passive outputs from a given system was characterised.

In this work, the connection between IDA-PBC and CbI for under-actuated mechanical systems is explored. Using the bond graph formalism (see \cite{Gawthrop2007} for introduction), a control sub-system is proposed that allows for shaping of the kinetic and potential energies of the closed-loop system. By representing the controller as a passive interconnection, the requisite matching conditions are reformulated in terms of the added mass and added potential energy. Equivalence between the CbI and IDA-PBC is then established in the case of mechanical systems by identifying Casimirs relating the controller states to those of the plant. Finally, using the reformulated matching conditions, it is shown that in the case that the mass matrix depends on only one coordinate that the kinetic energy matching conditions can be formulated as an ODE that can be evaluated numerically for implementation.

%
%
%

\noindent {\bf Notation.} Function arguments are declared upon definition and are omitted for subsequent use. $0_{n\times m}$ denotes a $n\times m$ zeros matrix whereas $I_n$ denotes a $n\times n$ identity matrix. For mappings $\mathcal{H}:\mathbb{R}^n \to \mathbb{R}$, we denote the transposed gradient as $\nabla\mathcal{H}:= \left(\frac{\partial \mathcal{H}}{\partial x}\right)^\top$. For $P = P^\top \in \mathbb{R}^{n \times n}$, $\lambda_{min}\left[P\right], \lambda_{max}\left[P\right]$ denotes the minimum and maximum (real) eigenvalues of P, respectively.

\section{BACKGROUND AND PROBLEM FORMULATION}\label{sec:background}
In this section a number of key concepts necessary for the subsequent developments are briefly revised.

\subsection{Control-by-interconnection}
In this work we consider input-state-output port-Hamiltonian systems (ISO-PHS) of the form
\begin{equation}\label{cbi:plant}
	\begin{split}
		\dot x_p
		&=
		F_p(x_p)\nabla_{x_p} H_p(x_p) + G_p(x_p)u_p \\
		y_p 
		&=
		G_p^\top(x_p)\nabla_{x_p} H_p(x_p)
	\end{split}
\end{equation}
where $x_p\in\mathbb{R}^p$ is the state of the plant, $F_p(x_p)\in\mathbb{R}^{p\times p}$ is the combined interconnection and damping matrix satisfying $F_p(x_p) + F_p^\top(x_p) \leq 0$, $H_p(x_p) \in\mathbb{R}$ is the Hamiltonian, $u_p\in\mathbb{R}^m$ is the input, $G_p(x_p)\in\mathbb{R}^{p\times m}$ is the input mapping matrix and $y_p\in\mathbb{R}^m$ is the natural passive output corresponding to the input $u_p$.

CbI assumes that the controller is a passive system that is interconnected with the plant \eqref{cbi:plant} via a passive interconnection. In this work, we consider a controller subsystem with two input-output ports, described by the ISO-PHS
\begin{equation}\label{cbi:controller}
		\begin{bmatrix}
			\dot x_c \\
			-y_{c1} \\
			-y_{c2}
		\end{bmatrix}
		=
		\underbrace{
		\begin{bmatrix}
			K_{11}(x_c) & K_{12}(x_c) & K_{13}(x_c) \\
			K_{21}(x_c) & K_{22}(x_c) & K_{23}(x_c) \\
			K_{31}(x_c) & K_{32}(x_c) & K_{33}(x_c)
		\end{bmatrix}
		}_{:= K(x_c)}
		\begin{bmatrix}
			\nabla_{x_c} H_c \\
			u_{c1} \\
			u_{c2}
		\end{bmatrix}
\end{equation}
where $x_c\in\mathbb{R}^c$ is the state of the controller, $H_c(x_c)\in\mathbb{R}$ is the controller Hamiltonian, $u_{c1}, y_{c1}\in\mathbb{R}^m$ and $u_{c2}, y_{c2}\in\mathbb{R}^r$ are passive input-output pairs and $K(x_c)\in\mathbb{R}^{(p+m+r)\times(p+m+r)}$ satisfies $K(x_c) + K^\top(x_c) \leq 0$ \cite{Schaft2017}.

The controller system \eqref{cbi:controller} can be interconnected with the plant \eqref{cbi:plant} via the passive interconnection
\begin{equation}
	\begin{split}
		u_p &= -y_{c1} \\
		u_{c1} &= y_p,
	\end{split}
\end{equation}
resulting in the closed-loop dynamics
\small
\begin{equation}\label{cbi:closedLoop}
	\begin{split}
		\begin{bmatrix}
			\dot x_p \\
			\dot x_c \\
			-y_{c2}
		\end{bmatrix}
		&=
		\underbrace{
		\begin{bmatrix}
			F_p+G_pK_{22}G_p^\top & G_pK_{21} & G_pK_{23} \\
			K_{12}G_p^\top & K_{11} & K_{13} \\
			K_{32}G_p^\top & K_{31} & K_{33} \\
		\end{bmatrix}
		}_{F_{cl}}
		\begin{bmatrix}
			\nabla_{x_p} H_p \\
			\nabla_{x_c} H_c \\
			u_{c2}
		\end{bmatrix}
	\end{split}
\end{equation}
\normalsize
where $u_{c2}, y_{c2}$ is a passive input-output pair to the interconnected system.
Noting that $K + K^\top \leq 0$, the closed-loop interconnection and damping structure $F_{cl}$ satisfies $F_{cl} + F_{cl}^\top \leq 0$ also.

In the case of stabilisation, the objective is to construct the plant functions $H_c(x_c)$, $K(x_c)$ to ensure the existence of Casimirs which statically relate the controller states to functions of the plant states
\begin{equation}\label{cbi:constraint}
	x_c = f_c(x_p),
\end{equation}
for $f_c(x)\in\mathbb{R}^c$. The Casimir functions and controller initial conditions are then designed to assign a desirable minimum to the total energy function
\begin{equation}
	W(x_p) = H(x_p) + H_c(x_c)|_{x_c = f_c(x_p)}.
\end{equation}
It is noted that the Lyapunov candidate $W(x_p)$ can be generalised to a function of $H, H_c$ and the Casimirs $x_c - f_c(x_p)$ \cite{Schaft2017}. Methods to ensure the existence of and constructing Casimirs have been reported in \cite{Ortega2008} and the references therein.

\subsection{Underactuated mechanical systems}
The primary objective of this work is to apply CbI to the class of underactuated mechanical systems, described by the dynamics
\begin{equation}\label{OLsystem}
	\begin{split}
		\begin{bmatrix}
			\dot q \\
			\dot p
		\end{bmatrix}
		&=
		\begin{bmatrix}
			0_{n\times n} & I_n \\
			-I_n & 0_{n\times n}
		\end{bmatrix}
		\begin{bmatrix}
			\nabla_q H \\
			\nabla_p H
		\end{bmatrix}
		+
		\begin{bmatrix}
			0_{n\times m} \\ G
		\end{bmatrix}
		u \\
		H(q,p)
		&=
		\underbrace{\frac12 p^\top M^{-1}(q)p}_{:=T(q,p)} + V(q) \\
		y
		&=
		G^\top\nabla_p H,
	\end{split}
\end{equation}
where $q\in\mathbb{R}^n, p\in\mathbb{R}^n$ are the configuration and momentum vectors, respectively, $u\in\mathbb{R}^m$ in the input, $M(q) = M^\top(q) > 0$ is the inertia matrix and $y$ is the natural passive output corresponding the the input $u$. The input mapping matrix $G$ is assumed to be constant and have the structure
\begin{equation}\label{Gdef}
	G
	=
	\begin{bmatrix}
		I_m \\ 0_{(n-m)\times m}
	\end{bmatrix},
\end{equation}
where $n-m < n$ is the degree of underactuation of the system\footnote{The assumed structure of $G$ requires that the first $m$ configuration coordinates are chosen to be collocated with the actuators. This class of dynamics falls into the broader class of ISO-PHS \eqref{cbi:plant}. For a more general input mapping matrix $\bar G(q)\in\mathbb{R}^{n\times m}$, there exists a change of coordinates recovering the structure \eqref{Gdef} if the columns of $\bar G(q)$ are involute.}. The Hamiltonian $H(q,p)$ is the sum of the kinetic energy $T(q,p)$ and the potential energy $V(q)$, which allows the gradient of $H$ with respect to $q$ to be written as
\begin{equation}
	\nabla_q H(q,p)
	=
	\nabla_q T(q,p) + \nabla_q V(q).
\end{equation}
A full-rank left-annihilator for the input mapping matrix \eqref{Gdef} is defined as
\begin{equation}
	G^\perp
	=
	\begin{bmatrix}
		0_{(n-m)\times m} & I_{(n-m)}
	\end{bmatrix}
\end{equation}
which satisfies $G^\perp G = 0_{(n-m)\times m}$.

In the subsequent development, we will require an alternate representation of the gradient of the kinetic energy with respect to configuration $\nabla_q T(q,p)$. Noting that the kinetic energy is quadratic in $p$, the gradient $\nabla_q T(q,p)$ can always be factored into the form
\begin{equation}
	\nabla_q T(q,p) = E(q,p)M^{-1}(q)p,
\end{equation}
for some matrix $E(q,p)\in\mathbb{R}^{n\times n}$. This has been previously noted in \cite{Reyes-Baez2017} using the Christoffel  symbols. Note, however, that the matrix $E(q,p)$ is non-unique and in this work we will use the representation given by
\begin{equation}\label{Edef}
	\begin{split}
		\nabla_q T(q,p)
		&=
		\frac12\frac{\partial^\top}{\partial q}\left(M^{-1}(q)p\right)p \\
		&=
		\underbrace{\frac12\frac{\partial^\top}{\partial q}\left(M^{-1}(q)p\right)M(q)}_{:=E(q,p)}M^{-1}(q)p.
	\end{split}
\end{equation}

In constructing a CbI interpretation to total energy shaping it is useful to define a virtual input-output pair for the system \eqref{OLsystem} by defining the input
\begin{equation}\label{uvDef}
	u_v = Gu,
\end{equation}
which allows the system to written similarly to a fully-actuated system as
\begin{equation}\label{OLsystem_alt}
	\begin{split}
		\begin{bmatrix}
			\dot q \\
			\dot p
		\end{bmatrix}
		&=
		\begin{bmatrix}
			0_{n\times n} & I_n \\
			-I_n & 0_{n\times n}
		\end{bmatrix}
		\begin{bmatrix}
			\nabla_q H \\
			\nabla_p H
		\end{bmatrix}
		+
		\begin{bmatrix}
			0_{n\times n} \\ I_n
		\end{bmatrix}
		u_v \\
		y_v
		&=
		\nabla_p H,
	\end{split}
\end{equation}
where $u_v, y_v\in\mathbb{R}^n$.
From the definition \eqref{uvDef}, it is clear that any input $u_v$ must satisfy
\begin{equation}\label{matchingCondition}
	G^\perp u_v = 0_{m\times 1},
\end{equation}
which will be ensured in subsequent control design. Assuming that \eqref{matchingCondition} holds, the input $u$ can be described as a function of $u_v$ by
\begin{equation}\label{uDef_uv}
	\begin{split}
		u
		&=
		G^\top u_v.
	\end{split}
\end{equation}
The advantage of constructing the virtual input-output pair is that the virtual output now describes the full velocity vector
\begin{equation}
	y_v = M^{-1}(q)p = \dot q,
\end{equation}
a property that will be exploited when shaping the potential energy.

\subsection{IDA-PBC for underactuated mechanical systems}
IDA-PBC is a control design methodology whereby the control signal is designed such that the closed-loop dynamics have a port-Hamiltonian (pH) structure. When applied to underactuated mechanical systems, the target closed-loop dynamics have the structure
\begin{equation}\label{IDAPBCsystem}
	\begin{split}
		\begin{bmatrix}
			\dot q \\
			\dot p
		\end{bmatrix}
		&=
		\begin{bmatrix}
			0_{n\times n} & M^{-1}(q)M_d(q) \\
			-M_d(q)M^{-1}(q) & J_2(q,p)-GK_dG^\top
		\end{bmatrix}
		\begin{bmatrix}
			\nabla_q H_d \\
			\nabla_p H_d
		\end{bmatrix} \\
		H_d&(q,p)
		=
		\frac12 p^\top M_d^{-1}(q)p + V_d(q),
	\end{split}
\end{equation}
where $M_d(q) = M_d^\top(q) > 0, V_d(q)$ are the desired closed-loop inertia matrix and potential energies, respectively, $J_2(q,p) = -J_2^\top(q,p)$ is skew-symmetric and $K_d = K_d^\top \geq 0$ is a tuning parameter used for damping injection. If $V_d$ is minimised at the target configuration, $H_d$ qualifies as a Lyapunov function for the closed-loop system \cite{Acosta2005}.

The complexity of applying IDA-PBC to underactuated system is satisfying the so-called matching conditions. This conditions requires that the dynamics of the open-loop and closed-loop systems must agree on the spaces perpendicular to the control signal. The structure chosen fo the closed-loop system in \eqref{IDAPBCsystem} ensures that the dynamics of $q$ agree with \eqref{OLsystem}. Comparing the dynamics of $p$ results in the condition
\begin{equation}
	\begin{split}
		G^\perp 
		\left\lbrace
		\nabla_q H-M_d(q)M^{-1}(q)\nabla_q H_d - J_2(q,p)\nabla_p H_d
		\right\rbrace \\
		=
		0_{(n-m)\times 1},
	\end{split}
\end{equation}
which defines a PDE that should be solved for $M_d(q), V_d(q), J_2(q,p)$. Noting the structure of the Hamiltonians, this PDE can be separated into the components involving $p$, and those that do not by
\begin{equation}\label{matchingEquations}
	\begin{split}
		G^\perp 
		\left\lbrace
			\nabla_q T-M_d(q)M^{-1}(q)\nabla_q T_d - J_2(q,p)M_d^{-1}(q)p
		\right\rbrace \\
		=
		0_{(n-m)\times 1} \\
		G^\perp 
		\left\lbrace
			\nabla_q V-M_d(q)M^{-1}(q)\nabla_q V_d
		\right\rbrace \\
		=
		0_{(n-m)\times 1}.
	\end{split}
\end{equation}
These expressions are known as the kinetic energy and potential energy matching equations, respectively.

\subsection{Contributions}
The objective of this work is to construct a CbI interpretation of IDA-PBC when applied to underactuated mechanical systems. The contributions of this work are threefold:
\begin{enumerate}[label=\textbf{C.\arabic*}]
\item ISO-PHS with Casimirs that statically relate states are considered and a closed-form solution to remove the Casimirs by reducing the dimension of the state vector is proposed. This solution can be applied to the closed-loop dynamics of CbI implementations of the form \eqref{cbi:closedLoop} to describe the resulting dynamics as a function of $x_p$ only.
\item A CbI controller for underactuated mechanical system of the form \eqref{cbi:controller} is proposed and the resulting closed-loop is shown to be equivalent to the well-known dynamics \eqref{IDAPBCsystem}. The CbI interpretation generates alternate matching conditions to the expressions \eqref{matchingEquations}, describing constraints on the added mass and added potential energy.
\item Using the alternate matching conditions, it is shown that the kinetic energy matching equations reduce to ODEs in the special case of underactuation degree one where the mass matrix is a function of only one configuration coordinate. It is demonstrated that numerical methods can be utilised in such cases to avoid solving these expressions analytically.
\end{enumerate}

\subsection{Related works}
Significant attention has been given to solving the matching equations \eqref{matchingEquations} in recent decades. In \cite{Gomez-Estern2001}, \cite{Ortega2002} it was shown that is the system is under-actuated degree one and the mass matrix depends on a single un-actuated coordinate, the kinetic energy matching condition can be simplified to an ODE. Using a novel parametrisation of $J_2(q,p)$, a general solution for under-actuated degree one system was proposed in \cite{Acosta2005} under the assumption that the mass matrix depends only on the actuated coordinates. This approach was extended in \cite{Viola2007} using a momentum transformation to simplify the matching equations. More recently, solutions to the potential energy matching equations were considered in \cite{Ryalat2016} under the assumption that the mass matrix and potential energy functions were dependent on only one variable. Finally, a general solution to the special case of 2 degree-of-freedom system was proposed in \cite{Arpenti2022}. The studies \cite{Gomez-Estern2004}, \cite{Sandoval2010} considered the effects of friction on the closed-loop stability using IDA-PBC.

In recent works, alternate approaches to constructing solutions to the matching equations have been explored. The existence of conservative forces that cannot be factorised into a skew-symmetric matrix $J_2(q,p)$ was investigated in \cite{Donaire2016b}, which resulted in alternate matching equations. Implicit system representations were used in \cite{Cieza2019} to construct solutions in an over-parameterised space where the closed-loop dynamics were subject to constraints. By working in the larger dimension, a solution to a under-actuated degree 2 crane system was proposed. Pfaffian differential equations were utilised in \cite{Harandi2021} which resulted in the kinetic energy PDEs being converted to an alternate form which admits simpler solutions.

Some authors have investigated the possibility of avoiding the matching equations altogether by considering the control signal to be a CbI. The work \cite{Donaire2016c} relied on a Lagrangian structure and several technical assumptions to verify the existence of a second passive output corresponding to the input $u$. Using this second output, a stabilising control was designed that ensured stability without requiring a solution to the matching PDEs. A similar approach was proposed in \cite{Romero2017a,Romero2016} where a second passive output was utilised and the control assumed to have a PID structure.

\section{CASIMIR REDUCTION}\label{sec:Casimir}
In this section, a method for reducing the state dimension of ISO-PHS with Casimirs is derived. The reduction method applies to general ISO-PHS with Casimirs and can be directly applied to the resulting closed-loop dynamics of CbI schemes of the form \eqref{cbi:closedLoop} to describe the system as a function of $x_p$ only. In the sequel, the reduction method will be used to show equivalence between the CbI controller for underactuated mechanical systems and the IDA-PBC dynamics \eqref{IDAPBCsystem}. Before introducing the state reduction solution, a useful lemma is required.

\begin{lemma}\label{lem:schurNegDef}
	Consider a square block matrix of arbitrary dimension
	\begin{equation}
		A
		=
		\begin{bmatrix}
			A_{11} & A_{12} \\
			A_{21} & A_{22}
		\end{bmatrix}
	\end{equation}
	and assume that $A_{22}$ is invertible.
	If the symmetric component of $A$ is negative semi-definite, $A + A^\top \leq 0$, the symmetric component of the Schur complement 
	\begin{equation}
		A_{11} - A_{12}A_{22}^{-1}A_{21}
	\end{equation}
	is negative semi-definite also.
\end{lemma}

\begin{proof}
	First note that the Schur complement of $X$ can be computed by
	\begin{equation}
		\begin{split}
			\begin{bmatrix}
				I & -A_{21}^\top A_{22}^{-\top}
			\end{bmatrix}
			A
			\begin{bmatrix}
				I \\ -A_{22}^{-1}A_{21}
			\end{bmatrix}
			=
			A_{11} - A_{12}A_{22}^{-1}A_{21}.
		\end{split}
	\end{equation}
	The symmetric component of this expression is negative semi-definite as $A + A^\top \leq 0$.
\end{proof}

The solution for reducing the dimension of ISO-PHS which exhibit Casimirs is now introduced. This development applies to systems of the form
\begin{equation}\label{Casimir:fullSystem}
	\begin{split}
		\begin{bmatrix}
			\dot x_1 \\ 
			\dot x_2 \\
			-y
		\end{bmatrix}
		&=
		\underbrace{
		\begin{bmatrix}
			F_{11}(x) & F_{12}(x) & F_{13}(x) \\
			F_{21}(x) & F_{22}(x) & F_{23}(x) \\
			F_{31}(x) & F_{32}(x) & F_{33}(x)
		\end{bmatrix}
		}_{F(x)}
		\begin{bmatrix}
			\nabla_{x_1} H \\
			\nabla_{x_2} H \\
			u
		\end{bmatrix}
	\end{split}
\end{equation}
where $x\in\mathbb{R}^{p+c}$ is the state of the system which has been partitioned into $x_1\in\mathbb{R}^p, x_2\in\mathbb{R}^c$, $H(x_1,x_2)\in\mathbb{R}$ is the Hamiltonian, $F(x)\in\mathbb{R}^{(p+c)\times(p+c)}$ is the full-rank interconnection and damping matrix satisfying $F(x) + F^\top(x) \leq 0$, $u\in\mathbb{R}^m$ is the input and $y\in\mathbb{R}^m$ is the corresponding passive output. It is assumed that the system contains a Casimir and the states have been partitioned such that the Casimir can be written as
\begin{equation}\label{CasimirDef}
	x_2 = f_c(x_1),
\end{equation}
where $f_c(x_1)\in\mathbb{R}^c$ is differentiable.

The first step in constructing a minimal system representation is defining a new set of coordinates given by
\begin{equation}\label{Casimir:wDef}
	w = x_2 - f_c(x_1) = 0_{c\times 1},
\end{equation}
which is identically equal to zero by construction.
The system \eqref{Casimir:fullSystem} can be described in the coordinates $(x_1, w)$ by
\begin{equation}\label{systemTransformed}
	\begin{split}
		\begin{bmatrix}
			\dot x_1 \\
			-y \\
			\dot w
		\end{bmatrix}
		=&
		\begin{bmatrix}
			I_p & 0_{p\times c} & 0_{p\times m} \\
			0_{m\times p} & 0_{m\times c} & I_m \\
			-\frac{\partial f_c}{\partial x_1} & I_c & 0_{c\times m}
		\end{bmatrix}
		\begin{bmatrix}
			\dot x_1 \\
			\dot x_2 \\
			-y
		\end{bmatrix} \\
		=&
		\begin{bmatrix}
			I_p & 0_{p\times c} & 0_{p\times m} \\
			0_{m\times p} & 0_{m\times c} & I_m \\
			-\frac{\partial f_c}{\partial x_1} & I_c & 0_{c\times m}
		\end{bmatrix}
		F \\
		&\times
		\begin{bmatrix}
			I_p & 0_{p\times m} & -\frac{\partial^\top f_c}{\partial x_1} \\
			0_{c\times p} & 0_{c\times m} & I_c \\
			0_{m\times p} & I_m & 0_{m\times c}
		\end{bmatrix}
		\begin{bmatrix}
			\nabla_{x_1} H_r \\
			u \\
			\nabla_{w} H_r
		\end{bmatrix} \\
		=&
		\underbrace{
		\begin{bmatrix}
			\bar F_{11} & \bar F_{12} & \bar F_{13} \\
			\bar F_{21} & \bar F_{22} & \bar F_{23} \\
			\bar F_{31} & \bar F_{32} & \bar F_{33}
		\end{bmatrix}
		}_{\bar F}
		\begin{bmatrix}
			\nabla_{x_1} H_r \\
			u \\
			\nabla_{w} H_r
		\end{bmatrix}, \\
	\end{split}
\end{equation} 
where
\begin{equation}
	\begin{split}
		H_r(x_1, w) :=& H(x_1, w + f_c(x_1)) \\
		\bar F_{11}(x_1)
		=&
		F_{11}(x)|_{x_2 = f_c(x_1)} \\
		\bar F_{12}(x_1)
		=&
		F_{13}(x)|_{x_2 = f_c(x_1)} \\
		\bar F_{13}(x_1)
		=&
		F_{12}(x)-F_{11}(x)\frac{\partial^\top f_c}{\partial x_1}\bigg|_{x_2 = f_c(x_1)} \\
		\bar F_{21}(x_1)
		=&
		F_{31}(x)|_{x_2 = f_c(x_1)} \\
		\bar F_{22}(x_1)
		=&
		F_{33}(x)|_{x_2 = f_c(x_1)} \\
		\bar F_{23}(x_1)
		=&
		F_{32}(x)-F_{31}(x)\frac{\partial^\top f_c}{\partial x_1}\bigg|_{x_2 = f_c(x_1)} \\
		\bar F_{31}(x_1)
		=&
		F_{21}(x)-\frac{\partial f_c}{\partial x_1}F_{11}(x)\bigg|_{x_2 = f_c(x_1)} \\
		\bar F_{32}(x_1)
		=&
		F_{23}(x)-\frac{\partial f_c}{\partial x_1}F_{13}(x)\bigg|_{x_2 = f_c(x_1)} \\
		\bar F_{33}(x_1)
		=&
		F_{22}(x)-F_{21}(x)\frac{\partial^\top f_c}{\partial x_1} - \frac{\partial f_c}{\partial x_1}F_{12}(x) \\
		&+ \frac{\partial f_c}{\partial x_1}F_{11}(x)\frac{\partial^\top f_c}{\partial x_1}\bigg|_{x_2 = f_c(x_1)}.
	\end{split}
\end{equation}

Recalling that $w$ is identically equal to zero, $\dot w$ is also equal to zero. Consequently, the final row of \eqref{systemTransformed} is a constraint that needs to be resolved to construct a minimal system representation. There are two methods of reduction that will be considered. Firstly, in the transformed coordinates \eqref{systemTransformed} it can occur that one or more columns of $\bar F_{\star3}$ is identically zero. Without loss of generality, it is assumed that the first $d$ columns of $\bar F_{\star3}$ are equal to zero. To remove the zero rows, the full-rank matrix $B$ is defined as
\begin{equation}\label{Bdef}
	B
	=
	\begin{bmatrix}
		0_{d\times (c-d)} \\
		I_{(c-d)}
	\end{bmatrix},
\end{equation}
which acts to select the non-zero columns of $\bar F_{\star3}$. The zero columns and corresponding rows are removed from \eqref{systemTransformed} by
\begin{equation}\label{systemTransformedReduced}
	\begin{split}
		\begin{bmatrix}
			\dot x_1 \\
			-y \\
			B^\top\dot w
		\end{bmatrix}
		=&
		\underbrace{
		\begin{bmatrix}
			\bar F_{11} & \bar F_{12} & \bar F_{13}B \\
			\bar F_{21} & \bar F_{22} & \bar F_{23}B \\
			B^\top\bar F_{31} & B^\top\bar F_{32} & B^\top\bar F_{33}B
		\end{bmatrix}
		}_{:= \bar F_B}
		\begin{bmatrix}
			\nabla_{x_1} H_r \\
			u \\
			B^\top\nabla_{w} H_r
		\end{bmatrix},
	\end{split}
\end{equation} 
which does not modify the system dynamics. Note that as $\bar F + \bar F^\top \leq 0$, $\bar F_B + \bar F_B^\top \leq 0$ also. Using this representation, a method for resolving the remaining constraint equations is presented under the assumption that $B^\top\bar F_{33}B$ is full rank.

\begin{proposition}\label{prop:casimir}
	Consider the pH system \eqref{Casimir:fullSystem} with Casimir \eqref{CasimirDef}. If the matrix $B^\top\bar F_{33}B$
	is full-rank for all $x_1$ the system can be described by a reduced-order model
	\begin{equation}\label{casimir:ReducedDynamics}
		\begin{split}
			\begin{bmatrix}
				\dot x_1 \\
				-y
			\end{bmatrix}
			&=
			F_r(x_1)
			\begin{bmatrix}
				\nabla_{x_1}H_r \\
				u
			\end{bmatrix},
		\end{split}
	\end{equation}
	where
	\begin{equation}\label{casimir:reducedMats}
		\begin{split}
			H_r(x_1)
			=&
			H\left(x_1, f_c(x_1)\right) \\
			F_r(x_1)
			=&
			\begin{bmatrix}
				\bar F_{11} & \bar F_{12} \\
				\bar F_{21} & \bar F_{22} \\
			\end{bmatrix} \\
			&-
			\begin{bmatrix}
				\bar F_{13}B \\
				\bar F_{23}B \\
			\end{bmatrix}
			(B^\top\bar F_{33}B)^{-1}
			\begin{bmatrix}
				B^\top\bar F_{31} & B^\top\bar F_{32}
			\end{bmatrix}
		\end{split}
	\end{equation}
	and $\bar F(x_1)$ satisfies $F_r(x_1) + F_r^\top(x_1) \leq 0$.
\end{proposition}

\begin{proof}
	The expression $B^\top\dot w$, defined in \eqref{systemTransformedReduced}, is identically equal to $0_{(c-d)\times 1}$ by construction. Note, however, that the gradient $B^\top\nabla_{w} H_r$ is not necessarily equal to zero. Assuming that $B^\top\bar F_{33}B$ is full rank, the expression $B^\top \nabla_w H_r$ can be described as
	\begin{equation}
		\begin{split}
			B^\top \nabla_{w} H_r
			=&
			-(B^\top\bar F_{33}B)^{-1}
			\begin{bmatrix}
				B^\top\bar F_{31} & B^\top\bar F_{32}
			\end{bmatrix}
			\begin{bmatrix}
				\nabla_{x_1} H_r \\
				u
			\end{bmatrix}
		\end{split}
	\end{equation}
	Substituting this expression into the dynamics \eqref{systemTransformedReduced} resolves to the reduced dynamics \eqref{casimir:ReducedDynamics}.
	
	To verify that $F_r + F_r^\top \leq 0$, note that $F_r$ is the Schur complement of $\bar F_B$ which satisfies $\bar F_B + \bar F_B^\top \leq 0$. It follows that $F_r + F_r^\top \leq 0$ by application of Lemma \ref{lem:schurNegDef}.
\end{proof}

Proposition \ref{prop:casimir} showed that an ISO-PHS that exhibits a Casimir function can be described in a reduced state-space. The class of dynamics that are derived from application of CbI \eqref{cbi:closedLoop} that result in a Casimir of the form \eqref{cbi:constraint} falls into the class of systems \eqref{Casimir:fullSystem}. The following Corollary tailors the Casimir reduction for this important sub-class of dynamics.

\begin{corollary}\label{corr:casimir}
	If the closed-loop dynamics of a CbI scheme \eqref{cbi:closedLoop} exhibit a Casimir of the form \eqref{cbi:constraint}, the system can be equivalently expressed in the form \eqref{casimir:ReducedDynamics} where
	\begin{equation}\label{FbarDefs}
		\begin{split}
			x_1
			=&
			x_p \\
			H_{r}(x_p)
			=&
			H_p(x_p) + H_c\left(f_c(x_p)\right) \\
			\bar F_{11}
			=&
			F_p+G_pK_{22}G_p^\top \\
			\bar F_{12}
			=&
			G_pK_{23} \\
			\bar F_{13}
			=&
			G_pK_{21}-\left[F_p+G_pK_{22}G_p^\top\right]\frac{\partial^\top f_c}{\partial x_p} \\
			\bar F_{21}
			=&
			K_{32}G_p^\top \\
			\bar F_{22}
			=&
			K_{33} \\
			\bar F_{23}
			=&
			K_{31}-K_{32}G_p^\top\frac{\partial^\top f_c}{\partial x_p} \\
			\bar F_{31}
			=&
			K_{12}G_p^\top-\frac{\partial f_c}{\partial x_p}\left[F_p+G_pK_{22}G_p^\top\right] \\
			\bar F_{32}
			=&
			K_{13}-\frac{\partial f_c}{\partial x_p}G_pK_{23} \\
			\bar F_{33}
			=&
			K_{11}-K_{12}G_p^\top\frac{\partial^\top f_c}{\partial x_p} - \frac{\partial f_c}{\partial x_p}G_pK_{21} \\
			&+ \frac{\partial f_c}{\partial x_p}\left[F_p+G_pK_{22}G_p^\top\right]\frac{\partial^\top f_c}{\partial x_p}
		\end{split}
	\end{equation}
	and $B$ is suitably chosen as per \eqref{Bdef} using the expressions $\bar F_{\star3}$.
	The arguments have been dropped from the definitions of $\bar F_{\star\star}(x_p)$ for the sake of readability.
\end{corollary}

\begin{proof}
	The result follows from direct application of Proposition \ref{prop:casimir} to the dynamics \eqref{cbi:closedLoop}.
\end{proof}

\section{CONTROL-BY-INTERCONNECTION FOR MECHANICAL SYSTEMS}\label{sec:CbI}
In this section, a control-by-interconnection scheme for under-actuated mechanical systems is presented. A dynamic 2-port control system is introduced with the intention that it will be interconnected to the plant \eqref{OLsystem_alt} via one of the ports. The controller states are constructed to be statically related to the plant states after interconnection, resulting in Casimirs. By applying Proposition \ref{prop:casimir}, the closed-loop dynamic are defined in a reduced space in which the dynamics coincide with standard total energy-shaping control \eqref{IDAPBCsystem}.

The proposed CbI scheme is shown in Figure \ref{underactuatedMassShaping}. The intention of this control subsystem is to interconnect with the plant \eqref{OLsystem_alt} via the $u_{c1}, y_{c1}$ power port. 
\begin{figure}[h!]
	\centering
	\includegraphics[width=1.0\columnwidth]{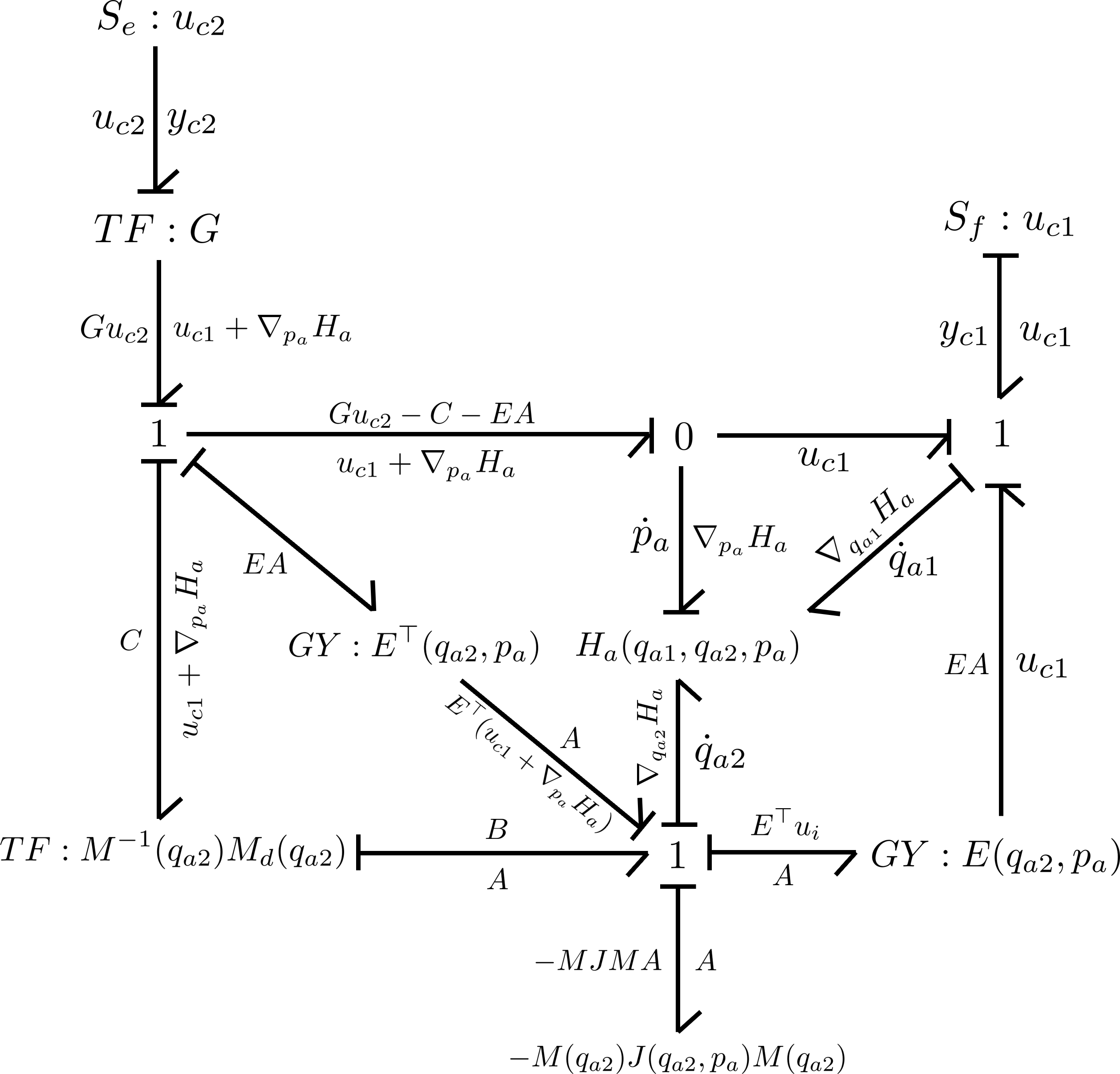}
	\caption{Total energy shaping as a CbI for under-actuated mechanical systems.}
	\label{underactuatedMassShaping}
\end{figure}
The second input $u_{c2}$ is available for subsequent control design, such as damping injection. The terms $M(q_{a_2}), E(q_{a2},p_a)$ are the plant mass matrix \eqref{OLsystem} and factorisation of the kinetic energy gradient \eqref{Edef} evaluated at the controller states whereas $J(q_{a2},p_a) = -J(q_{a2},p_a)^\top\in\mathbb{R}^{n\times n}$ is a skew-symmetric matrix to be chosen. The three-port storage element $H_a(q_{a1},q_{a2},p_a)$ has states $q_{a1}, q_{a2}, p_a\in\mathbb{R}^n$ and energy function similar to mechanical systems,
\begin{equation}\label{HaDef}
	\begin{split}
		H_a(q_{a1},q_{a2},p_a)
		&=
		\underbrace{\frac12 p_a^\top M_a^{-1}(q_{a2})p_a}_{:=T_a(q_{a2},p_a)} + \underbrace{V_d(q_{a2}) - V(q_{a1})}_{:=V_a(q_{a1},q_{a2})},
	\end{split}
\end{equation}
where $M_a^{-1}(q_{a2})$ is the inverse added mass, $T_a(q_{a2},p_a)$ is the added kinetic energy, $V_d(q_{a2})$ is the desired closed-loop potential energy, $V(q_{a1})$ is the plant potential energy function \eqref{OLsystem} evaluated at the plant state $q_{a1}$ and $V_a(q_{a1},q_{a2})$ is the total added potential energy. It is important to note that, although $M_a^{-1}(q_{a2})$ is represented as a matrix inverse, it need not be invertible nor positive. Indeed, it will be shown in subsequent developments that the key requirement is that 
\begin{equation}\label{MdDef}
	M_d^{-1}(q) := M^{-1}(q)+M_a^{-1}(q)
\end{equation}
should be positive definite.

In subsequent analysis it will be shown that the interconnection of the control system with the plant \eqref{OLsystem_alt} via the interconnection
\begin{equation}\label{mechInterconnection}
	\begin{split}
		u_v &= -y_{c1} \\
		u_{c1} &= y_v,
	\end{split}
\end{equation}
yields Casimirs
\begin{equation}\label{mechCasimir}
	\begin{bmatrix}
		q_{a1} \\ q_{a2} \\ p_a
	\end{bmatrix}
	=
	\underbrace{
	\begin{bmatrix}
		I_n & 0_{n\times n} \\
		I_n & 0_{n\times n} \\
		0_{n\times n} & I_n
	\end{bmatrix}
	\begin{bmatrix}
		q \\ p
	\end{bmatrix}
	}_{f_c(x_p)}.
\end{equation}
Assuming the Casimirs exist, some intuition regarding the construction of the control system in Figure \ref{underactuatedMassShaping} can be provided. Both $q_{a1}, q_{a2}$ were constructed to be equal to $q$. Firstly $\dot q_{a1}$ is equal to $\dot q$ by interconnection with the plant virtual output $y_v$ via a 1-junction. To verify a similar relation for $q_{a2}$, assume that $q_{a2} = q, p_{a} = p$ holds which results in $\nabla_{p_a}H_a = M_a^{-1}(q)p$ and $u_{c1}+\nabla_{p_a}H_a = M_d^{-1}p$. With this in mind, the transformer can be seen to reconstruct the velocity $\dot q = M^{-1}(q)p$ for the bottom 1-junction, resulting in $\dot q_{a1} = \dot q$.

To construct a Casimir $p_a = p$, first note from \eqref{OLsystem} and \eqref{Edef} that the plant momentum dynamics can be expressed as
\begin{equation}
	\dot p = -\nabla_{q}V - E(q,p)M^{-1}(q)p + u_v.
\end{equation}
The control structure acts to remove these forces from the plant via the right side of the control structure and re-introduce them via the top 0-junction where they are shared with the dynamics of $p_a$. The $\dot q_{a1}$ bond acts to cancel the gravity term from the plant $-\nabla_{q}V$. Recalling that the bottom 1-junction has flow equal to $M^{-1}(q)p$, the right-side gyrator cancels the term $- E(q,p)M^{-1}(q)p$ from the plant. The left-side gyrator then re-introduces the force $-E(q,p)M^{-1}(q)p$ via the top 0-junction where it is shared between $\dot p$ and $\dot p_a$, establishing the desired Casimir.

The claimed Casimir \eqref{mechCasimir} is now formalised in the following Proposition. For this development, note that the gradients of the added energy $H_a(\cdot)$ satisfy
\begin{equation}\label{HaGrad}
	\begin{split}
		\nabla_{q_{a1}}H_a
		&=
		-\nabla_{q_{a1}}V \\
		\nabla_{q_{a2}}H_a
		&=
		\nabla_{q_{a2}}T_a + \nabla_{q_{a2}}V_d \\
		\nabla_{q_{a2}}T_a
		&=
		\frac12\frac{\partial^\top}{\partial q_{a2}}\left(M_a^{-1}(q_{a2})p_a\right)p_a \\
		\nabla_{p_a}H_a
		&=
		M_a^{-1}(q_{a2})p_a
	\end{split}
\end{equation}
and the expressions $A(\cdot), B(\cdot), C(\cdot)$ in Figure \ref{underactuatedMassShaping} can be evaluated as
\begin{equation}\label{ABCdefs}
	\begin{split}
		A(q_{a2},p_{a},u_{c1})
		=&
		M^{-1}(q_{a2})M_d(q_{a2})\left[u_{c1} + \nabla_{p_a}H_a\right] \\
		B(q_{a2},p_{a},u_{c1})
		=&
		\nabla_{q_{a2}}H_a - E^\top(q_{a2},p_{a})\nabla_{p_a}H_a \\
		&- M(q_{a2})J(q_{a2},p_a)M_d(q_{a2}) \\
		&\times\left[u_{c1} + \nabla_{p_a}H_a\right] \\
		C(q_{a2},p_{a},u_{c1})
		=&
		M_d(q_{a2})M^{-1}(q_{a2})\nabla_{q_{a2}}H_a \\
		&- M_d(q_{a2})M^{-1}(q_{a2})E^\top(q_{a2},p_{a})\nabla_{p_a}H_a \\
		&- M_d(q_{a2})J(q_{a2},p_a)M_d(q_{a2}) \\
		&\times\left[u_{c1} + \nabla_{p_a}H_a\right],
	\end{split}
\end{equation}
with $M_d(\cdot)$ defined in \eqref{MdDef}.
To ensure that the Casimir exists, a number of requirements are imposed on the selection of the added inverse mass $M_a^{-1}(q)$ and closed-loop potential energy $V_d(q)$ which are equivalent of the standard matching conditions used in IDA-PBC \eqref{matchingEquations}.

\begin{proposition}\label{prop:Casimir}
	Consider the control system in Figure \ref{underactuatedMassShaping} and assume that it is interconnected to the plant \eqref{OLsystem_alt} via the interconnection \eqref{mechInterconnection}. If $M_a^{-1}(q_a), V_d(q_a)$ are chosen such that
	\begin{equation}\label{matching}
		G^\perp C(q_{a2},p_a,u_{c1})|_{q_{a2}=q,p_a=p} = G^\perp \nabla_q V
	\end{equation}
	and the controller states are initialised as $q_{a1}(0) = q_{a2}(0) = q(0)$, $p_a(0) = p(0)$, the Casimir \eqref{mechCasimir} holds for all time.
\end{proposition}

\begin{proof}
	Consider that at some time instant $T$ \eqref{mechCasimir} holds, implying that
	\begin{equation}
		\begin{split}
			q_{a1}(T) = q_{a2}(T) = q(T), \ p_a(T) = p(T).
		\end{split}
	\end{equation}
	It is shown that if \eqref{matching} is satisfied, then the derivatives of the states also agree 
	\begin{equation}\label{CasimirDerivatives}
		\begin{split}
			\dot q_{a1}(T) = \dot q_{a2}(T) = \dot q(T), \ \dot p_a(T) = \dot p(T),
		\end{split}
	\end{equation}
	establishing the existence of a Casimir for all future time.
	
	We proceed by first establishing the relationship for the configuration vector. From \eqref{mechInterconnection} and \eqref{OLsystem_alt} $u_c = M^{-1}(q)p$ which establishes $\dot q_{a1}(T) = \dot q(T)$. The input $u_c = M^{-1}(q)p$ is substituted into $A(\cdot)$ \eqref{ABCdefs} to find
	\begin{equation}\label{Adef}
		\begin{split}
			A|_{t=T} &= M^{-1}(q_{a2})M_d(q_{a2})\left[M^{-1}(q)p + M_{a}^{-1}(q_{a2})p_a\right]|_{t=T} \\
			&= M^{-1}(q)p|_{t=T} \\
			&= \dot q|_{t=T},
		\end{split}
	\end{equation}
	confirming that $\dot q_{a2}(T) = \dot q(T)$.
	
	Next we consider the behaviour of the momentum states. First note that, from the bond graph in Figure \ref{underactuatedMassShaping} and the definition \eqref{uDef_uv}, the plant input $u_v$ is given by
	\begin{equation}\label{controlDef}
		\begin{split}
			u &= G^\top u_v \\
			&= 
			G^\top\left[Gu_{c2} - C(q_{a2},p_a,u_{c1}) + \nabla_{q_{a1}} V(q_{a1})\right] \\
			&= 
			u_{c2} - G^\top\left[ C(q_{a2},p_a,u_{c1}) - \nabla_{q_{a1}} V(q_{a1})\right].
		\end{split}
	\end{equation}
	Using the control definition \eqref{controlDef} and the condition \eqref{matching}, the plant dynamics \eqref{OLsystem} can be expanded as
	\begin{equation}
		\begin{split}
			\dot p
			=&
			-\nabla_q T(q,p) - 
			\begin{bmatrix}
				G^\top \nabla_q V(q) \\
				G^\perp \nabla_q V(q)
			\end{bmatrix}
			+ Gu \\
			=&
			-\nabla_q T(q,p) - 
			\begin{bmatrix}
				G^\top \nabla_q V(q) \\
				G^\perp \nabla_q V(q)
			\end{bmatrix} \\
			&+ 
			G\left\lbrace
			u_{c2} - G^\top\left[ C(q_{a2},p_a,u_{c1}) - \nabla_{q_{a1}} V(q_{a1})\right]
			\right\rbrace \\
			=&
			-\nabla_q T(q,p) 
			+ 
			Gu_{c2} \\
			&- 
			\begin{bmatrix}
				G^\top \left\lbrace C(q_{a2},p_a,u_{c1}) + \nabla_q V(q) - \nabla_{q_{a1}} V(q_{a1}) \right\rbrace \\
				G^\perp \nabla_q V(q)
			\end{bmatrix} \\
		\end{split}
	\end{equation}
	Note that at time $T$, $\nabla_q V(q)|_{t=T} = \nabla_{q_{a1}} V(q_{a1})|_{t=T}$. Additionally recall the assumption \eqref{matching} which allows the simplification
	\begin{equation}\label{pDef_timeT}
		\begin{split}
			\dot p|_{t=T}
			=&
			-\nabla_q T(q,p) - 
			C(q_{a2},p_a,u_{c1})
			+ 
			Gu_{c2}. \\
		\end{split}
	\end{equation}
	Recalling the identity \eqref{Adef}, the dynamics of $p_a$ at time $T$ can be expanded to
	\begin{equation}\label{paDef_timeT}
		\begin{split}
			\dot p_a
			&=
			Gv - E(q_{a2},p_a)A(\cdot)|_{t=T} - C(q_{a2},p_a,u_{c1})|_{t=T} \\
			&=
			Gv - \nabla_q T(q,p)|_{t=T} - C(q_{a2},p_a,u_{c1})|_{t=T},
		\end{split}
	\end{equation}
	which agrees with \eqref{pDef_timeT}.
	As \eqref{pDef_timeT} and \eqref{paDef_timeT} agree at time $T$, \eqref{CasimirDerivatives} is verified for the momentum states. If at the initial time $t=0$ we have $q_a(0) = q(0)$, $p_a(0) = p(0)$, it follows that $q_a(t) = q(t)$ and $p_a(t) = p(t)$ for all time via integration, completing the proof.
\end{proof}

Proposition \ref{prop:Casimir} has established that the Casimir \eqref{mechCasimir} holds under some technical assumptions that will be verified in subsequent design. Before proceeding, we note that the control subsystem in Figure \ref{underactuatedMassShaping} can be written in the form \eqref{cbi:controller} with
\small
\begin{equation}\label{controlMatDefs}
	\begin{split}
		x_c
		&=
		\begin{bmatrix}
			q_{a1}^\top & q_{a2}^\top & p_a^\top
		\end{bmatrix}^\top \\
		H_c(q_{a1},q_{a2},p_a) &= H_a(q_{a1},q_{a2},p_a) \\	
		K_{11}(q_{a2},p_a)
		&=
		\begin{bmatrix}
			0_{n\times n} & 0_{n\times n} & 0_{n\times n} \\
			0_{n\times n} & 0_{n\times n} & M^{-1}M_d \\
			0_{n\times n} & -M_dM^{-1} & D - D^\top + M_dJM_d
		\end{bmatrix} \\
		K_{12}(q_{a2},p_a)
		&=
		\begin{bmatrix}
			I_n \\
			M^{-1}M_d \\
			M_dJM_d - D^\top
		\end{bmatrix} \\
		K_{13}
		&=
		\begin{bmatrix}
			0_{n\times m} \\ 0_{n\times m} \\ G	
		\end{bmatrix} \\
		K_{21}(q_{a2},p_a)
		&=
		\begin{bmatrix}
			-I_n &- M_dM^{-1} & D + M_dJM_d \\
		\end{bmatrix} \\
		K_{31}
		&=
		\begin{bmatrix}
			0_{m\times n} & 0_{m\times n} & -G^\top
		\end{bmatrix} \\
		K_{22}(q_{a2},p_a)
		&=
		M_dJM_d \\
		K_{23}
		&=
		G \\
		K_{32}
		&=
		-G^\top \\
		K_{33}
		&=
		0_{m\times m} \\
		D(q_{a2},p_a)
		&=
		M_dM^{-1}E^\top.
	\end{split}
\end{equation}
\normalsize

In the subsequent developments it is assumed that the requisite \eqref{matching} of Proposition \ref{prop:Casimir} holds, implying $q_{a1}(t) = q_{a2}(t) = q(t)$, $p_a(t) = p(t)$. Condition \eqref{matching} will be verified by choice of $M_a^{-1}$ and $V_d$. Assuming the Casimir holds, the expressions for $A(\cdot), B(\cdot), C(\cdot)$ in \eqref{ABCdefs} can be simplified to
\begin{equation}\label{ABCdefs_proj}
	\begin{split}
		A(q,p) 
		=&
		M^{-1}(q)p \\
		B(q,p) 
		=&
		\nabla_{q}T_a(q,p) + \nabla_{q}V_d(q,p) - E^\top(q,p)M_a^{-1}(q)p \\ 
		&- M(q)J(q,p)p  \\
		C(q,p) 
		=&
		M_d(q)
		\left\lbrace
			M^{-1}(q)\nabla_{q}T_a(q,p) + M^{-1}(q)\nabla_{q}V_d(q,p) \right. \\
			&\left.- M^{-1}(q)E^\top(q,p)M_a^{-1}(q)p - M^{-1}(q)J(q,p)
		\right\rbrace  \\
	\end{split}
\end{equation}

Recalling the definition of $\nabla_{q_a}T_a$ in \eqref{HaGrad}, it is noted that $C(\cdot)$ contains some terms which are quadratic in $p$ and some that are functions of $q$ only. The function $C(\cdot)$ is divided into
\begin{equation}\label{Cdivision}
	\begin{split}
		C(q,p) 
		=&
		C_{KE}(q,p) + C_{PE}(q) \\
		C_{KE}(q,p)
		=&
		\underbrace{
		\left[M_a^{-1}(q) + M^{-1}(q)\right]^{-1}
		}_{M_d(q)}
		\left\lbrace
			Y(q,p) - J(q,p)
		\right\rbrace p \\
		C_{PE}(q)
		&=
		\left[M_a^{-1}(q) + M^{-1}(q)\right]^{-1}
		M^{-1}(q)\nabla_{q}V_d,
	\end{split}
\end{equation}
where $KE$ represents \textit{kinetic energy}, $PE$ represents \textit{potential energy} and $Y$ is defined as
\begin{equation}\label{XYdefs}
	\begin{split}
		Y(q,p)
		=&
		\frac12M^{-1}(q)\frac{\partial^\top}{\partial q}\left(M_a^{-1}(q)p\right) \\
		&- \frac12\frac{\partial}{\partial q}\left(M^{-1}(q)p\right)M_a^{-1}(q).
	\end{split}
\end{equation}
As $Y$ is linear in $p$ it can be written as
\begin{equation}
	\begin{split}
		Y(q,p)
		=
		\sum_{i=1}^n p_i Y^i(q),
	\end{split}
\end{equation}
where
\begin{equation}\label{YiDef}
	\begin{split}
		Y^i(q)
		=&
		\frac12 M^{-1}(q)\frac{\partial^\top}{\partial q}\left(M_a^{-1}(q)e_i\right) \\
		&- \frac12\frac{\partial}{\partial q}\left(M^{-1}(q)e_i\right)M_a^{-1}(q).
	\end{split}
\end{equation}

The key constraint for control design is choosing $M_a^{-1}(q), V_d(q)$ satisfying the matching condition \eqref{matching}. From the definition of $C(\cdot)$ in \eqref{ABCdefs_proj}, the constraint equation is a function of both $M_a^{-1}(q)$ and $\left[M_a^{-1}(q) + M^{-1}(q)\right]^{-1}$, making direct design of this matrix difficult. To simplify the design process, an alternate characterisation of \eqref{matching} is introduced.

In the following proposition, the inverse mass matrix, inverse added mass matrix, interconnection matrix and $Y(\cdot)$ are partitioned as
\begin{equation}
	\begin{split}
		\begin{bmatrix}
			m_{11}(q) & m_{21}^\top(q) \\
			m_{21}(q) & m_{22}(q)
		\end{bmatrix} 
		&=
		\begin{bmatrix}
			G^\top \\ G^\perp
		\end{bmatrix}
		M^{-1}(q)
		\begin{bmatrix}
		G & G^{\perp\top}
		\end{bmatrix} \\
		\begin{bmatrix}
			m_{a11}(q) & m_{a21}^\top(q) \\
			m_{a21}(q) & m_{a22}(q)
		\end{bmatrix} 
		&=
		\begin{bmatrix}
			G^\top \\ G^\perp
		\end{bmatrix}
		M_a^{-1}(q)
		\begin{bmatrix}
		G & G^{\perp\top}
		\end{bmatrix} \\
		\begin{bmatrix}
			J_{11}(q,p) & -J_{21}^\top(q,p) \\
			J_{21}(q,p) & J_{22}(q,p)
		\end{bmatrix}
		&=
		\begin{bmatrix}
			G^\top \\ G^\perp
		\end{bmatrix}
		J(q,p)
		\begin{bmatrix}
		G & G^{\perp\top}
		\end{bmatrix} \\
		\begin{bmatrix}
			Y_{11}(q,p) & Y_{12}(q,p) \\
			Y_{21}(q,p) & Y_{22}(q,p)
		\end{bmatrix}
		&=
		\begin{bmatrix}
			G^\top \\ G^\perp
		\end{bmatrix}
		Y(q,p)
		\begin{bmatrix}
		G & G^{\perp\top}
		\end{bmatrix}.
	\end{split}
\end{equation}
Using the above definitions, an alternate characterisation of \eqref{matching} is presented.

\begin{proposition}\label{prop:matchingEquations}
	The matching condition \eqref{matching} is satisfied if:
	\begin{itemize}
	\item The added mass matrix $M_a^{-1}(q)$ is chosen such that
	\begin{equation}\label{matching:KE}
		\begin{split}
			D(q)
			\left[Y^i(q) + Y^{i\top}(q)\right]
			D^\top(q) = 0_{(n-m)\times(n-m)}
		\end{split}
	\end{equation}
	for all $i\in\left\lbrace1,\dots,n\right\rbrace$ where
	\begin{equation}\label{Ddef}
		D(q)
		=
		\begin{bmatrix}
			(m_{21}+m_{a21})(m_{11}+m_{a11})^{-1} & -I_{n-m}
		\end{bmatrix}.
	\end{equation}
	
	\item The desired potential energy $V_d(q)$ satisfies
	\begin{equation}\label{matching:PE}
		\begin{split}
			s_1(q)G^\perp\nabla_q V
			&=
			-s_2(q)G^\top \nabla_{q}V_d - s_3(q)G^\perp \nabla_{q}V_d \\
			&=
			-D(q)M^{-1}(q)\nabla_q V_d,
		\end{split}
	\end{equation}
	where
	\begin{equation}\label{s1Def}
		\begin{split}
			s_1&(q) = (m_{22}+m_{a22}) \\
			&-(m_{21}+m_{a21})(m_{11}+m_{a11})^{-1}(m_{21}^\top+m_{a21}^\top)
		\end{split}
	\end{equation}
	is the Schur complement of $M^{-1}(q)+M_a^{-1}(q)$ and
	\begin{equation}\label{s2s3Def}
		\begin{split}
			s_2(q)
			&=
			(m_{21}+m_{a21})(m_{11}+m_{a11})^{-1}m_{11}-m_{21} \\
			s_3(q)
			&=
			(m_{21}+m_{a21})(m_{11}+m_{a11})^{-1}m_{21}^\top-m_{22}.
		\end{split}
	\end{equation}
	\end{itemize}
\end{proposition}

\begin{proof}
From the graph in Figure \ref{underactuatedMassShaping} and the interconnection \eqref{mechInterconnection}, the virtual input $u_v$ is given by
\begin{equation}\label{matching2}
	\begin{split}
		u_v
		=
		Gu
		=&
		Gv - C + \nabla_q V.
	\end{split}
\end{equation}
Recalling the definition of $G$, \eqref{matching2} is equivalent to \eqref{matching}.
Collecting the terms $u, v$ and left multiplying by $M_a^{-1} + M^{-1}$ results in
\begin{equation}\label{altMatching1}
	\begin{split}
		&\left[M_a^{-1} + M^{-1}\right]G(u-v) \\
		&\phantom{--}=
		\left[M_a^{-1} + M^{-1}\right]
		\left\lbrace
		- C_{KE} - C_{PE} + \nabla_q V
		\right\rbrace
	\end{split}
\end{equation}
Due to the structure of $G$ in \eqref{Gdef}, \eqref{altMatching1} has the left annihilator $D(\cdot)$, defined in \eqref{Ddef}.

Left multiplying \eqref{altMatching1} by $D(\cdot)$ and separating the components into those relating to the kinetic and potential energies result in
\begin{align}
	0_{(n-m)\times 1}
	&=
	-D
	\left[M_a^{-1} + M^{-1}\right]
	C_{KE} \label{altMatching:KE1} \\
	0_{(n-m)\times 1}
	&=
	-D
	\left[M_a^{-1} + M^{-1}\right]
	\left\lbrace
	C_{PE} - \nabla_q V
	\right\rbrace. \label{altMatching:PE1}
\end{align}
Using \eqref{Cdivision}, \eqref{altMatching:PE1} is expanded to
\begin{equation}
	\begin{split}
		0_{(n-m)\times 1}
		=&
		D
		\left[M_a^{-1} + M^{-1}\right]\nabla_q V \\
		&-
		DM^{-1}\nabla_{q}V_a,
	\end{split}
\end{equation}
which can be seen to agree with \eqref{matching:PE} after expanding.

Now considering the constraint on the kinetic energy expression \eqref{altMatching:KE1}, the definition \eqref{Cdivision} is substituted to find
\begin{equation}\label{altMatching:KE3}
	\begin{split}
		0_{(n-m)\times n}
		=&
		D
		\left\lbrace
			- Y + J
		\right\rbrace. \\
	\end{split}
\end{equation}
Using the relevant definitions, the first component of \eqref{altMatching:KE3} can be solved for $J_{21}(q,p)$ as
\begin{equation}
	\begin{split}
		J_{21}(q,p)
		=&
		(m_{21}+m_{a21})(m_{11}+m_{a11})^{-1}(-Y_{11}+J_{11})\\
		&+Y_{21}.
	\end{split}
\end{equation}
Substituting this expression back into the second component of \eqref{altMatching:KE3} reveals the constraint
\begin{equation}\label{altMatching:KE5}
	\begin{split}
		0&_{(n-m)\times (n-m)} \\
		=&
		(m_{21}+m_{a21})(m_{11}+m_{a11})^{-1}(- Y_{12} - J_{21}^\top) \\
		&\phantom{=}- (- Y_{22} + J_{22}) \\
		=&(m_{21}+m_{a21})(m_{11}+m_{a11})^{-1}\left[- Y_{12} - Y_{21}^\top\right] \\
		&- (m_{21}+m_{a21})(m_{11}+m_{a11})^{-1}(-Y_{11}+J_{11})^\top \\
		&\times(m_{11}+m_{a11})^{-1}(m_{21}^\top+m_{a21}^\top) - (- Y_{22} + J_{22}) \\
		=&
		-D
		\begin{bmatrix}
			-Y_{11}+J_{11}^\top & -Y_{12}-Y_{21}^\top \\
			0_{(n-m)\times m} & -Y_{22}+J_{22}
		\end{bmatrix}
		D^\top.
	\end{split}
\end{equation}
The term $J_{22}$ is taken as below to solve the skew-symmetric part of this expression,
\begin{equation}
	\begin{split}
		&J_{22} \\
		&=
		-\frac12 D
		\begin{bmatrix}
			-Y_{11}+Y_{11}^\top+J_{11}^\top-J_{11} & -Y_{12}-Y_{21}^\top \\
			Y_{12}^\top+Y_{21} & -Y_{22}+Y_{22}^\top
		\end{bmatrix}
		D^\top,
	\end{split}
\end{equation}
where $J_{11}\in\mathbb{R}^{m\times m}$ is a free skew-symmetric term.
The symmetric part of \eqref{altMatching:KE5} must also be equal to zero, implying that
\begin{equation}
	\begin{split}
		D
		\left[Y + Y^{\top}\right]
		D^\top = 0_{(n-m)\times(n-m)}.
	\end{split}
\end{equation}
Finally, noting that this must be true for each $p_i$, the condition \eqref{matching:KE} follows.
\end{proof}

\begin{remark}\label{rem:resolvingPDEs}
	The expression \eqref{matching:KE} implicitly defines a set of PDEs that must be satisfied by any choice of $M_a^{-1}(q)$. From the definition of $Y^i$ in \eqref{YiDef}, the first $m$ equations are describe partial differential equations involving the partial derivatives of $m_{a11}, m_{a21}$. The remaining $n-m$ equations describe partial differential equations involving the partial derivatives of $m_{a21}, m_{a22}$. This structure can be useful for resolving the equations into a standard representation for solving.
\end{remark}


\begin{corollary}\label{corr:KE_ODE}
	In the special case of under-actuation degree 1, if $M^{-1}$, $M_a^{-1}$ is a function of only 1 configuration variable $q_i$, the kinetic energy matching equations \eqref{matching:KE} can be reduced to a set of ODEs.
	\begin{equation}\label{underactuated1ODE}
		\begin{split}
			\frac{d}{d q_i}
			\begin{bmatrix}
				m_{a21}^\top \\ m_{a22}
			\end{bmatrix}
			=
			g\left(m_{a11},\frac{d}{d q_i}m_{a11},M^{-1},\frac{d}{d q_i}M^{-1}\right),
		\end{split}
	\end{equation}
	where $g(\cdot)\in\mathbb{R}^{n}$ is a function implicitly defined by the matching conditions \eqref{matching:KE} and $m_{a11}(q_i)$ can be chosen freely.
\end{corollary}

\begin{proof}
	Assuming that the mass matrix $M^{-1}$ is function only of a single configuration variable $q_i$, we will also impose that the added mass $M_a^{-1}$ is a function only of the same variable. As a consequence, the matching expression \eqref{matching:KE} is now only a function in the single variable $q_i$. Notably, all partial derivatives of $M_a^{-1}$ with respect to $q_k$, where $k\neq i$, are equal to zero.
	
	Noting Remark \ref{rem:resolvingPDEs}, the first $n-1$ expressions of \eqref{matching:KE} produce differential equations involving the partial derivatives of $m_{a11}, m_{a21}$. The dimension of $m_{a21}$ is $1\times (n-1)$, so the first $n-1$ equations can be solved simultaneously to find an expression for $\frac{d}{d q_i}m_{a21}^\top$. The $n^{th}$ expressions of \eqref{matching:KE} can then be resolved for an expression for $\frac{d}{d q_i}m_{a22}$, which has dimension 1. Combining these expressions, the matching equations \eqref{matching:KE} can be resolved into an ODE of the form \eqref{underactuated1ODE}.
\end{proof}

\begin{remark}
	Corollary \ref{corr:KE_ODE} describes situations in which the kinetic energy matching equations can be reduced to an ODE. The solution, however, will depend on the choice of $m_{a11}(q_i)$ and may not be globally defined. This poses the question of how should the function $m_{a11}(q_i)$ be chosen to ensure an appropriate solution $M_a^{-1}(q_i)$---a nonlinear control problem!
\end{remark}

The results of Corollary 1 describe the degrees of freedom that exist when constructing a solution to the added inverse mass matrix. Similar degrees of freedom exist in the definition of the closed-loop potential energy that can be exploited to ensure positivity of the chosen function. The following Corollary defines a free function that can be utilised to this effect.

\begin{corollary}\label{cor:PotEnergyDecomp}
	Suppose that there exists a full rank matrix-valued function $K(q)\in\mathbb{R}^{m\times m}$ such that the integral
	\begin{equation}\label{GammaDef}
		\begin{split}
			\Gamma(q)
			&=
			\int K(q)G^\top\left[M_a^{-1}(q) + M^{-1}(q)\right]M(q) \ dq,
		\end{split}
	\end{equation}
	exists. The desired closed-loop potential energy can be chosen as
	\begin{equation}\label{VdDecomp}
		V_d(q)
		=
		V_{m}(q) + V_{f}(\Gamma(q)),
	\end{equation}
	where $V_m(\cdot)$ must be chosen to satisfy the potential energy matching conditions \eqref{matching:PE} and $V_{f}(\cdot)$ is a free function that does not impact the matching equations. Consequently, the matching equation \eqref{matching:PE} can be equivalently written as
	\begin{equation}\label{matching:PE_aug}
		\begin{split}
			s_1(q)G^\perp\nabla_q V
			&=
			-s_2(q)G^\top \nabla_{q}V_m - s_3(q)G^\perp \nabla_{q}V_m \\
			&=
			-D(q)M^{-1}(q)\nabla_q V_m.
		\end{split}
	\end{equation}
\end{corollary}

\begin{proof}
	Computing the gradient of $V_d$ results in
	\begin{equation}\label{nablaVa}
		\begin{split}
			\nabla_q V_d
			&=
			\nabla_q V_m
			+
			\frac{\partial^\top\Gamma}{\partial q}\nabla_\Gamma V_f \\
			&=
			\nabla_q V_m
			+
			M\left[M_a^{-1} + M^{-1}\right]GK^\top\nabla_\Gamma V_f.
		\end{split}
	\end{equation}
	From the definition of $D(q)$ in \eqref{Ddef}, we have the identity
	\begin{equation}
		\begin{split}
			D(q)M^{-1}M(q)\left[M_a^{-1}(q) + M^{-1}(q)\right]G
			&=
			\begin{bmatrix}
				I_{m} & \star
			\end{bmatrix}
			G \\
			&=
			0_{(n-m)\times m}
		\end{split}
	\end{equation}
	Substituting the expression \eqref{nablaVa} into \eqref{matching:PE} and noting the above expression results in the simplified matching equation \eqref{matching:PE_aug}. 
\end{proof}

\begin{remark}
	$V_f(\cdot)$ is a free function precisely because $\Gamma$ is an integral of the passive output $y_{c2}$. The potential energy $V_f$ could be alternatively constructed as a capacitor element added to the input $u_{c2}$ in Figure \ref{underactuatedMassShaping}.
\end{remark}

Now we arrive at one of the key results of this work, the equivalence of the proposed CbI scheme and total energy-shaping control of underactuated mechanical systems. Assuming that the CbI scheme has been constructed to satisfy the required matching conditions to ensure the existence of a Casimir of the form $q_a = q, p_a = p$, Proposition \ref{prop:casimir} is applied to reconstruct the reduced closed-loop structure \eqref{IDAPBCsystem}.

\begin{proposition}\label{Prop:stability}
	Consider the underactuated mechanical system with virtual input \eqref{OLsystem_alt} and assume that $M_a(q), V_d(q)$ are chosen such that the conditions of Proposition \ref{prop:matchingEquations} are satisfied in some neighbourhood of a point $(q,p) = (q^\star, 0_{n\times 1})$. If the control signal is chosen as
	\begin{equation}
		\begin{split}
			u(q,p)
			=&
			v - G^\top\left\lbrace
			M_d(q)M^{-1}(q)
			\left[
			-E^\top(q,p)M_a^{-1}(q)p\right.\right. \\ 
			&\left.\left.+ \nabla_{q_a}H_a(q_a,p_a)
			- M(q)J(q,p)p\right]
			- \nabla_q V\right\rbrace
		\end{split}
	\end{equation}
	where
	\begin{equation}
		M_d(q) = \left[M_a^{-1}(q) + M^{-1}(q)\right]^{-1},
	\end{equation}
	the following hold:
	\begin{itemize}
	\item The closed-loop dynamics have the form
	\begin{equation}\label{MechSystemReduced}
		\begin{split}
			\begin{bmatrix}
				\dot q \\
				\dot p
			\end{bmatrix}
			=&
			\begin{bmatrix}
				0_{n\times n} & M^{-1}(q)M_d(q) \\
				-M_d(q)M^{-1}(q) & J_2(q,p)
			\end{bmatrix}
			\begin{bmatrix}
				\nabla_q H_d \\
				\nabla_p H_d
			\end{bmatrix} \\
			&+
			\begin{bmatrix}
				0_{n\times m} \\
				G
			\end{bmatrix}
			v \\
			H_d(q,p&)
			=
			\frac12 p^\top M_d^{-1}(q)p + V_d(q) \\
			&y
			=
			G^\top\nabla_p H_d,
		\end{split}
	\end{equation}
	where
	\small
	\begin{equation}
		\begin{split}
			J_2(q,p)
			=&
			M_d(q)\left\lbrace J(q,p) + M^{-1}(q)\left[E(q,p) - E^\top(q,p)\right]\right.\\
			&\left.\times M^{-1}(q)\right\rbrace M_d(q) + M_d(q)M^{-1}(q)E^\top(q,p)\\
			& - E(q,p)M^{-1}(q)M_d(q)
		\end{split}
	\end{equation}
	\normalsize
	
	\item 
	If $M_d(q), V_d(q)$ satisfy
		\begin{equation}
			\begin{split}
				M_d(q) > 0, \ \ V_d(q) > 0
			\end{split}
		\end{equation}
		in some neighbourhood of $(q,p) = (q^\star, 0_{n\times 1})$, $(q^\star, 0_{n\times 1})$ is a stable equilibrium of the closed-loop system for $v = 0_{m\times 1}$.
		
		\item If the input signal $v$ is used for damping injection
		\begin{equation}
			v = -K_dG^\top y
		\end{equation}
		for some positive $K_d\in\mathbb{R}^{m\times m}$ and the equilibrium $(q,p) = (q^\star, 0_{n\times 1})$, $(q^\star, 0_{n\times 1})$ is locally detectable from the output $y$, the point $(q^\star, 0_{n\times 1})$ is asymptotically stable.
	\end{itemize} 
\end{proposition}

\begin{proof}
	Interconnection of the mechanical system with the control subsystem results in a closed-loop of the form \eqref{cbi:closedLoop}, where $x_c$, $H_c$ and $K_{\star\star}$ are defined in \eqref{controlMatDefs} and
	\begin{equation}
		\begin{split}
			x_p
			=&
			\begin{bmatrix}
				q \\ p
			\end{bmatrix} \\
			F_p
			=&
			\begin{bmatrix}
				0_{n\times n} & I_n \\
				-I_n & 0_{n\times n}
			\end{bmatrix} \\
			G_p
			=&
			\begin{bmatrix}
				0_{n\times n} \\ I_n
			\end{bmatrix}.
		\end{split}
	\end{equation}
	From \eqref{mechCasimir}, we have that
	\begin{equation}
		\frac{\partial f_c}{\partial x_p} 
		=
		\begin{bmatrix}
			I_n & 0_{n\times n} \\
			I_n & 0_{n\times n} \\
			0_{n\times n} & I_n
		\end{bmatrix}.
	\end{equation}
	To verify the claim, Corollary \ref{corr:casimir} is applied which requires a suitable definition of $B$. Expanding the definitions of $\bar F_{\star 3}$ from \eqref{FbarDefs} reveals
	\begin{equation}
		\begin{split}
			\bar F_{13}
			&=
			\begin{bmatrix}
				0_{n\times n} & 0_{n\times n} & -I_n \\
				0_{n\times n} & I_n-M_dM^{-1} & D
			\end{bmatrix} \\
			\bar F_{23}
			&=
			\begin{bmatrix}
				0_{n\times n} & 0_{n\times n} & 0_{n\times n}
			\end{bmatrix} \\
			\bar F_{33}
			&=
			\begin{bmatrix}
				0_{n\times n} & 0_{n\times n} & 0_{n\times n} \\
				0_{n\times n} & 0_{n\times n} & I_n \\
				0_{n\times n} & -I_n & 0_{n\times n} \\
			\end{bmatrix},
		\end{split}
	\end{equation}
	resulting in the choice
	\begin{equation}
		B
		=
		\begin{bmatrix}
			0_{n\times n} & 0_{n\times n} \\
			I_n & 0_{n\times n} \\
			0_{n\times n} & I_n \\
		\end{bmatrix}.
	\end{equation}
	Expanding the expression $B^\top\bar F_{33} B$ results in
	\begin{equation}
		B^\top\bar F_{33} B
		=
		\begin{bmatrix}
			0_{n\times n} & -I_n \\
			I_n & 0_{n\times n}
		\end{bmatrix}
	\end{equation}
	which is invertible, ensuring that Corollary \ref{corr:casimir} can be applied. Expanding the definitions of $F_r$ in \eqref{casimir:reducedMats} results in the reduced dynamics
	\begin{equation}
		\begin{split}
			\begin{bmatrix}
				\dot q \\ \dot p \\ -y
			\end{bmatrix}
			&=
			\underbrace{
			\begin{bmatrix}
				0_{n\times n} & M^{-1}M_d & 0_{n\times n} \\
				-M_dM^{-1} & \bar J_2 & G \\
				0_{n\times n} & -G^\top & 0_{n\times n}
			\end{bmatrix}
			}_{F_r}
			\begin{bmatrix}
				\nabla_q H_d \\ \nabla_p H_d \\ v	
			\end{bmatrix} \\
			\bar J_2
			&=
			M_dJM_d+D-D^\top+M_dM^{-1}D^\top-DM^{-1}M_d,
		\end{split}
	\end{equation}
	which agrees with \eqref{MechSystemReduced} when substituting in the definition for $D$ in \eqref{controlMatDefs}. Stability and asymptotic stability of the point $(q^\star, 0_{n\times 1})$ follows from Proposition 1 of \cite{Ortega2002}.
\end{proof}

\begin{remark}
	From Proposition \ref{Prop:stability} it is clear that $M_a^{-1}(q)$ does not need to be a positive matrix. Rather, the closed-loop mass $M_d^{-1}$ must be positive to ensure stability fo the system. In cases that $M_a^{-1}$ is positive, the control sub-system in Figure \ref{underactuatedMassShaping} is passive.
\end{remark}

\section{Example applications}\label{sec:examples}
In this section the matching conditions derived in Proposition \ref{prop:matchingEquations} are used to construct stabilising control laws for the cart-pole and acrobot systems. In both cases, the mass matrix depends on only one configuration variable, so the kinetic energy matching conditions can be reduced to ODEs as detailed in Corollary \ref{corr:KE_ODE}. This enables the solutions to be constructed numerically, removing the need to analytically solve the equations.

Both examples were prepared in Matlab 2022a and the source code is available via \href{https://github.com/JoelFerguson/Underactuated_Mechanical_CbI}{https://github.com/JoelFerguson/Underactuated\_Mechanical\_CbI}.

\subsection{Cart-pole example}\label{Sec:example:cartPole}
The cart-pole system, shown in Figure \ref{fig:cartPole}, attempts to balance the pole of length $\ell$ and mass $m_p$ in the upright position by applying a force $F$ to the cart with mass $m_c$. The state $q_1$ describes the horizontal displacement of the cart whereas $q_2$ describes the angle of the pole from vertical in the clockwise direction.
\begin{figure}[h!]
	\centering
	\includegraphics[width=0.5\columnwidth]{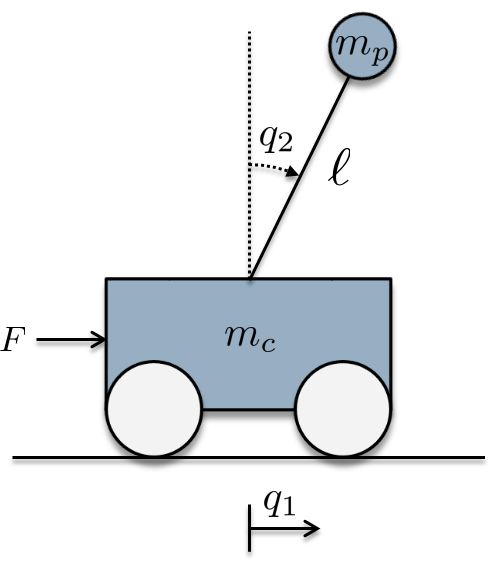}
	\caption{The cart-pole system attempts to balance the pole in the upright position by regulating the force $F$.}
	\label{fig:cartPole}
\end{figure}
The cart-pole system can be written as a pH system of the form \eqref{OLsystem} with
\begin{equation}\label{cartPole:model}
	\begin{split}
		q
		&=
		\begin{bmatrix}
			q_1 \\ q_2
		\end{bmatrix} \\
		M(q)
		&=
		\begin{bmatrix}
		m_c+m_p & m_p l\cos q_2 \\
		m_p l\cos q_2 & m_p l^2
		\end{bmatrix} \\
		V(q)
		&=
		m_p gl\cos q_2 \\
		G
		&=
		\begin{bmatrix}
			1 \\ 0
		\end{bmatrix}.
	\end{split}
\end{equation}
In the subsequent control design, the parameters $m_c = m_p = l = 1, g = 9.8$ have been used.

The mass matrix of the cart-pole system depends only on $q_2$, the unactuated coordinate. The added inverse mass is assumed to also be a function of $q_2$ also, allowing it to be written as
\begin{equation}
	M_a^{-1}(q_2)
	=
	\begin{bmatrix}
		m_{a11}(q_2) & m_{a21}^\top(q_2) \\
		m_{a21}(q_2) & m_{a22}(q_2)
	\end{bmatrix}.
\end{equation}
As noted in Corollary \ref{corr:KE_ODE}, the kinetic energy matching equations \eqref{matching:KE} can be reduced to an ODE as both $M^{-1}, M_a^{-1}$ are a function of only one variable. The associate ODE is of the form \eqref{underactuated1ODE} for $q_i = q_2$ where $m_{a11}(q_2)$ is a free function to be chosen. The ODE can be evaluated using numerical solvers.

Before solving the ODE associated with the kinetic energy matching equations, consideration should be given to how the resulting mass matrix impacts the closed-loop potential energy $V_d$. Recalling \eqref{VdDecomp}, the closed-loop potential energy is composed of a free term $\Gamma(\cdot)$ and a term $V_m(q)$ which must satisfy \eqref{matching:PE_aug}, where $s_1, s_2, s_3$ are defined in \eqref{s1Def}, \eqref{s2s3Def}. As the potential $V, M^{-1}, M_a^{-1}$ are all functions of only $q_2$, $V_m$ is also assumed to be a function of $q_2$ only, reducing \eqref{matching:PE_aug} to the ODE
\begin{equation}\label{cartPole:potentialMatching}
	\begin{split}
		\nabla_{q_2}V_m
		&=
		-\frac{s_1(q_2)}{s_3(q_2)}\nabla_{q_2} V,
	\end{split}
\end{equation}
which can be evaluated numerically once a solution for $M_a^{-1}(q_2)$, and hence $s_1(\cdot), s_3(\cdot)$, are found. noting that the vector field $\nabla_{q_2} V$ is divergent from the point $q_2 = 0$, the closed-loop vector field $\nabla_{q_2}V_m$ should reverse the direction locally. This is ensured if the ratio $\frac{s_1(q)}{s_3(q)}$ is positive in some neighbourhood of the origin. Recalling that $s_1(q)$ is the Schur complement of $M^{-1} + M_a^{-1}$, which is necessarily positive, it is required that $s_3(q)$ be positive in some neighbourhood of $q_2 = 0$. The values
\begin{equation}\label{massMatInitCond}
	\begin{split}
		m_{a11}(0) &= 0 \\
		m_{a21}(0) &= -2 \\
		m_{a22}(0) &= 8,
	\end{split}
\end{equation}
where chosen which result in $s_1(0) = 1$, $s_3(0) = 1$ and $\lambda_{min}\left[M^{-1}(0) + M_a^{-1}(0)\right] = 0.917 > 0$.

The added inverse mass matrix can now be found by numerically evaluating the ODE \eqref{underactuated1ODE}. The term $m_{a11}(q_2)$ is a free function that was chosen to be constant $m_{a11}(q_2) = 0, \frac{\partial}{\partial q_2}m_{a11} = 0$ for this example. The resulting functions for $m_{a21}(q_2), m_{a22}(q_2)$ were found to exist on the interval $q_2\in\left[-0.48, 0.48\right]$ and are shown in Figure \eqref{fig:addedMass}. From Proposition \ref{Prop:stability}, $M^{-1}(q_2) + M_a^{-1}(q_2)$ should be positive to ensure stability, so the minimum eigenvalue of this expression is shown in the same figure.
\begin{figure}[h!]
	\centering
	\includegraphics[width=1.0\columnwidth]{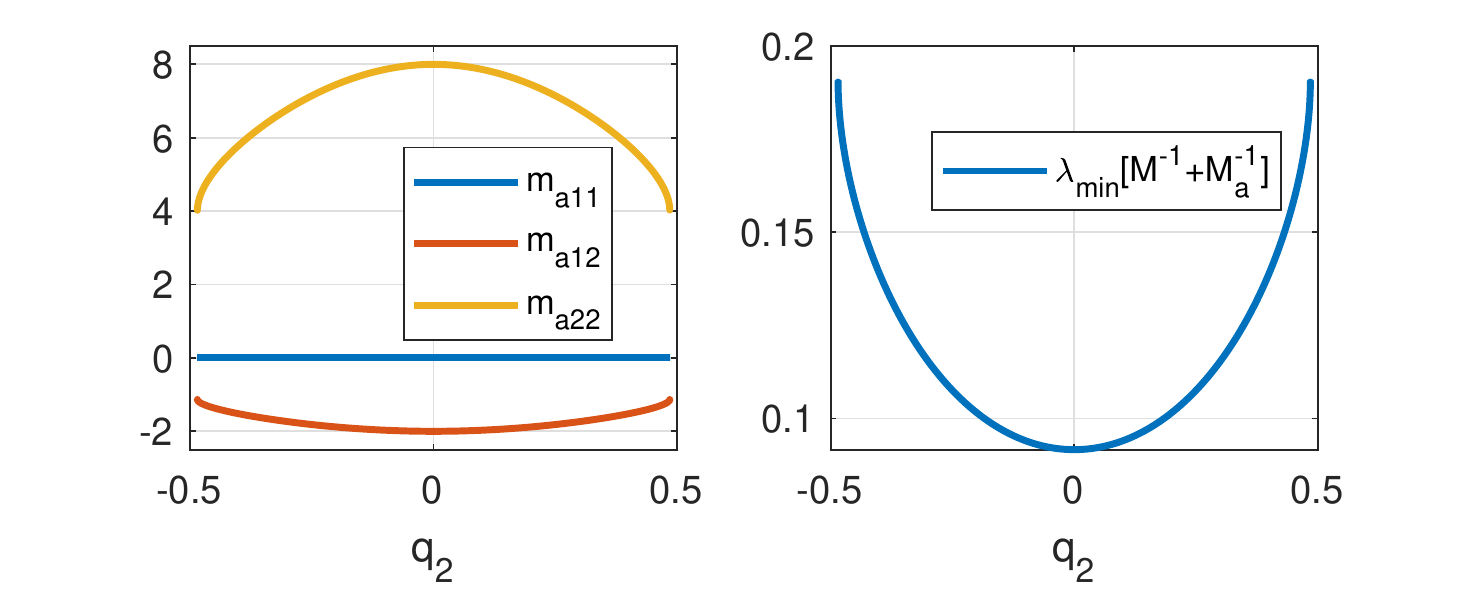}
	\caption{A solution for the inverse added mass $M_a^{-1}$ was found to exist for the cart-pole system on the domain $q_2$ between $\pm 0.48$ radians.}
	\label{fig:addedMass}
\end{figure}

The closed-loop potential energy $V_m(q_2)$ can now be obtained by numerically by evaluating the ODE \eqref{cartPole:potentialMatching}. The terms $s_1(\cdot), s_3(\cdot)$ are evaluated using the solutions to $M_a^{-1}$ shown in Figure \eqref{fig:addedMass}. The resulting function $V_m(q_2)$ is shown in Figure \eqref{fig:addedPotEnergy}. As expected, the function is positive in some neighbourhood of $q_2 = 0$ due to the choice of the added mass at $q_2 = 0$ in \eqref{massMatInitCond}.
\begin{figure}[h!]
	\centering
	\includegraphics[width=1.0\columnwidth]{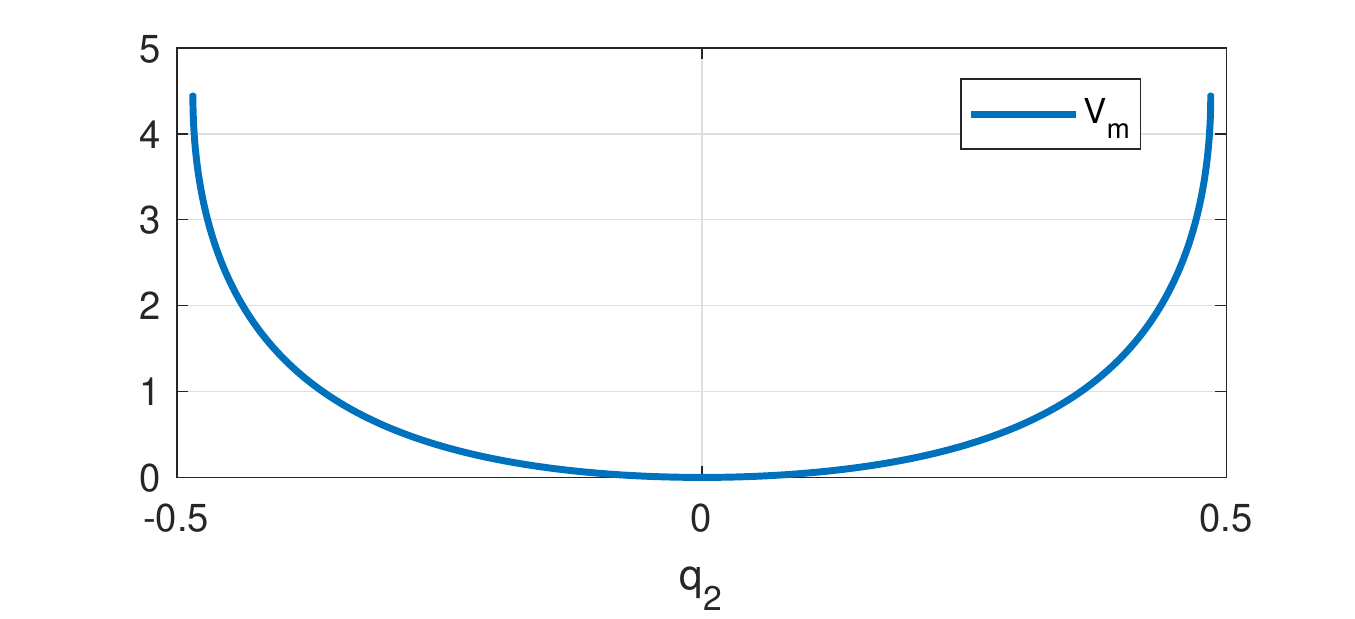}
	\caption{Solution for the added potential energy $V_m(q_2)$ for the cart-pole system.}
	\label{fig:addedPotEnergy}
\end{figure}

The proposed functions of $M_a^{-1}$, $V_m$ can be used to construct a controller to stabilise the pendulum in the upright position. To ensure stability of $q_1 = 0$ also, the free term $V_f(\Gamma(q))$, defined in \eqref{VdDecomp}, is constructed. The function $\Gamma(\cdot)$ defined by the integral \eqref{GammaDef}, where $K(q)$ is a free function chosen to ensure solvability. Noting that $M^{-1}, M_a^{-1}$ are functions of $q_2$ only, the parametrisation
\begin{equation}
	\begin{split}
		\begin{bmatrix}
			\beta_1(q_2) & \beta_2(q_2)
		\end{bmatrix}
		=
		G^\top\left[M_a^{-1}(q_2) + M^{-1}(q_2)\right]M(q_2)
	\end{split}
\end{equation}
is introduced. The free function is chosen as $K(q_2) = \frac{1}{\beta_1(q_2)}$, resulting in 
\begin{equation}
	\begin{split}
		\Gamma(q)
		&=
		\int
		\begin{bmatrix}
			1 & \frac{\beta_2(q_2)}{\beta_1(q_2)}
		\end{bmatrix}
		\ dq \\
		&=
		q_1 + \int \frac{\beta_2(q_2)}{\beta_1(q_2)} \ dq_2,
	\end{split}
\end{equation}
which can be solved numerically from the initial condition $\Gamma(0_{2\times 1}) = 0$. The function $V_f(\cdot)$ was taken as $V_f(\Gamma(q)) = \frac12\kappa\Gamma(q)^2$ with $\kappa = 5$ for simulation. A contour plot of the resulting closed-loop potential energy is shown in Figure \eqref{fig:addedPotEnergyVd}. Note that a minimum has been assigned to $q = 0_{2\times 1}$.
\begin{figure}[h!]
	\centering
	\includegraphics[width=1.0\columnwidth]{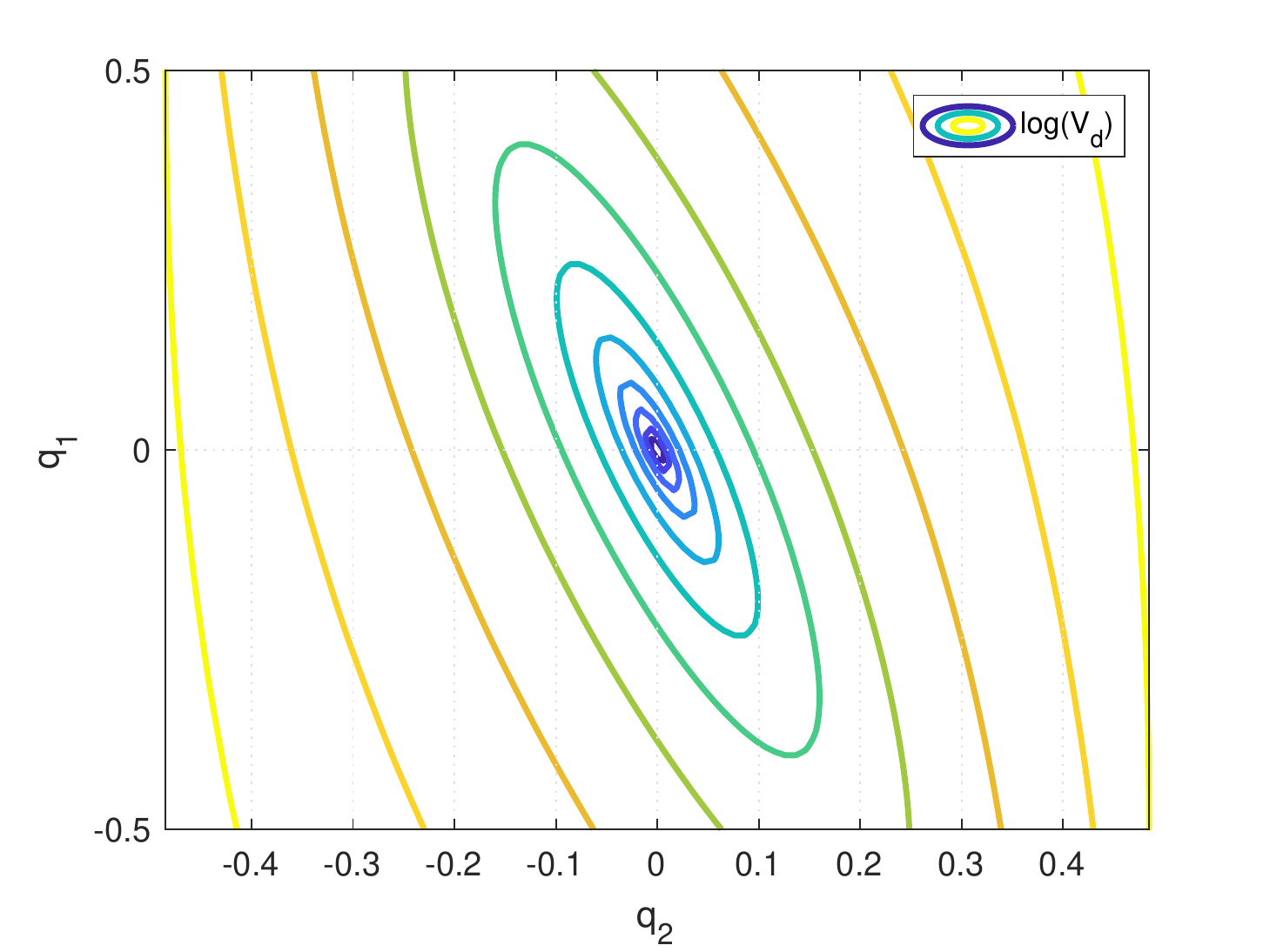}
	\caption{Contour plot of the closed-loop potential energy $V_d(q) = V_m(q_2) + V_f(\Gamma(q))$ for the cart-pole system on log scale.}
	\label{fig:addedPotEnergyVd}
\end{figure}
As a final control design stage, damping is injected via the new passive input/output pair with
\begin{equation}
	v
	=
	-5G^\top (M_a^{-1} + M^{-1})p.
\end{equation}
The complete control signal is defined by the expression \eqref{controlDef}.

The cart-pole system was simulated for 5 seconds from initial conditions $q(0) = (0, 0.3)$, $p(0) = (0,0)$. The resulting state evolution and closed-loop energy $H_d$ is shown in Figure \ref{fig:simulation}. As expected, the proposed controller stabilises the origin and the closed-loop energy $H_d$ decreases monotonically.
\begin{figure}[h!]
	\centering
	\includegraphics[width=1.0\columnwidth]{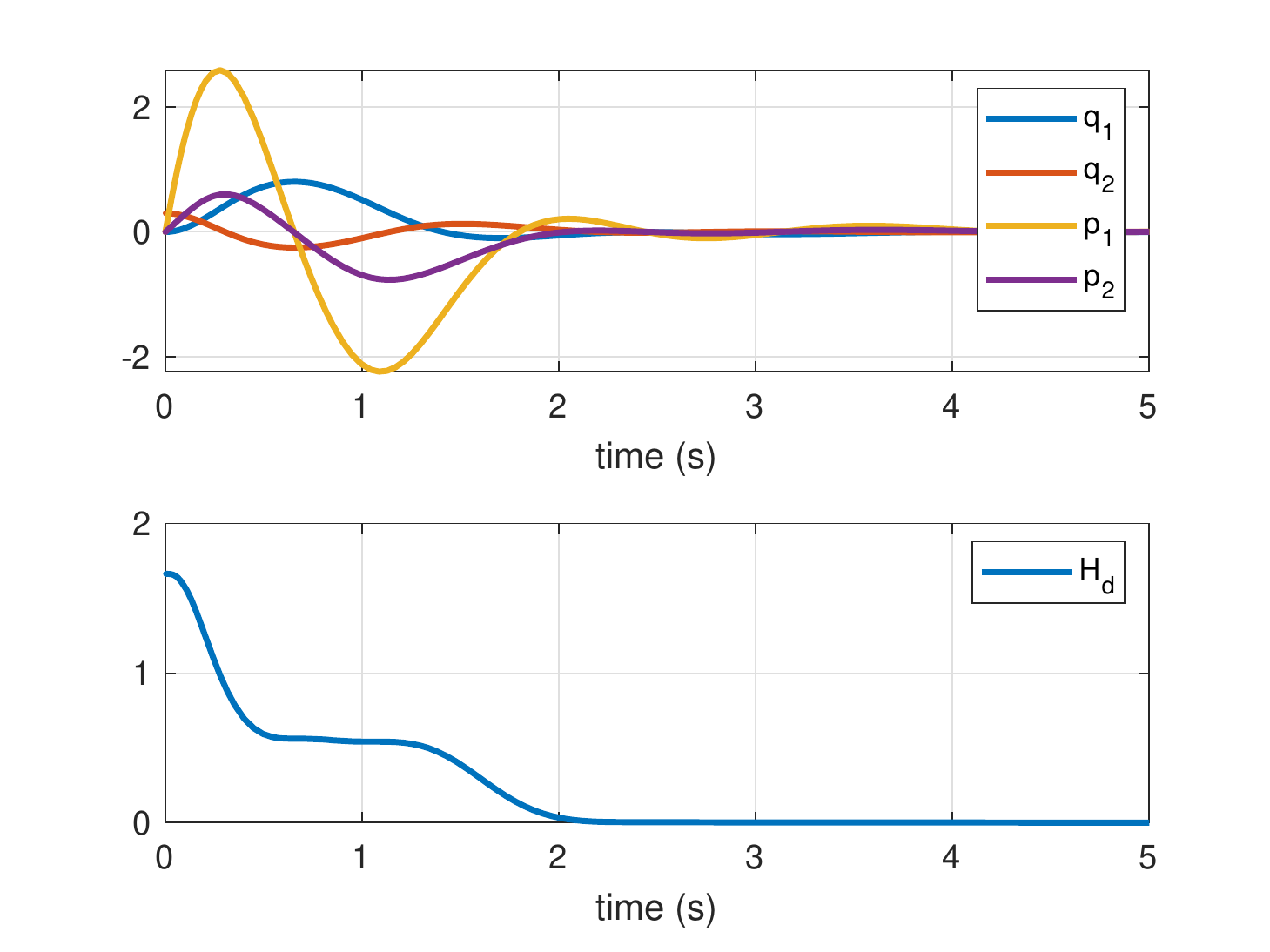}
	\caption{Numerical simulation of cart-pole system in closed-loop with CbI scheme.}
	\label{fig:simulation}
\end{figure}

\subsection{Acrobot example}\label{Sec:example:acrobot}
The acrobot system, shown in Figure \ref{fig:acrobot}, consists of 2 links with an actuator supplying a input torque $\tau$ fixed between the base and second links. The base link has displacement of $q_2$, measured from vertical, length $\ell_2$, mass $m_2$, moment of inertia $J_{\ell1}$ and centre of mass $\ell_{c2}$ from the base pivot point. The actuated link has displacement of $q_1$ measured relative to the base link, length $\ell_1$, mass $m_1$, moment of inertia $J_{\ell1}$ and centre of mass $\ell_{c1}$ from the actuated pivot point.  The control objective of this system is to stabilise the upright equilibrium position $(q_1,q_2) = (0,0)$.
\begin{figure}[h!]
	\centering
	\includegraphics[width=0.5\columnwidth]{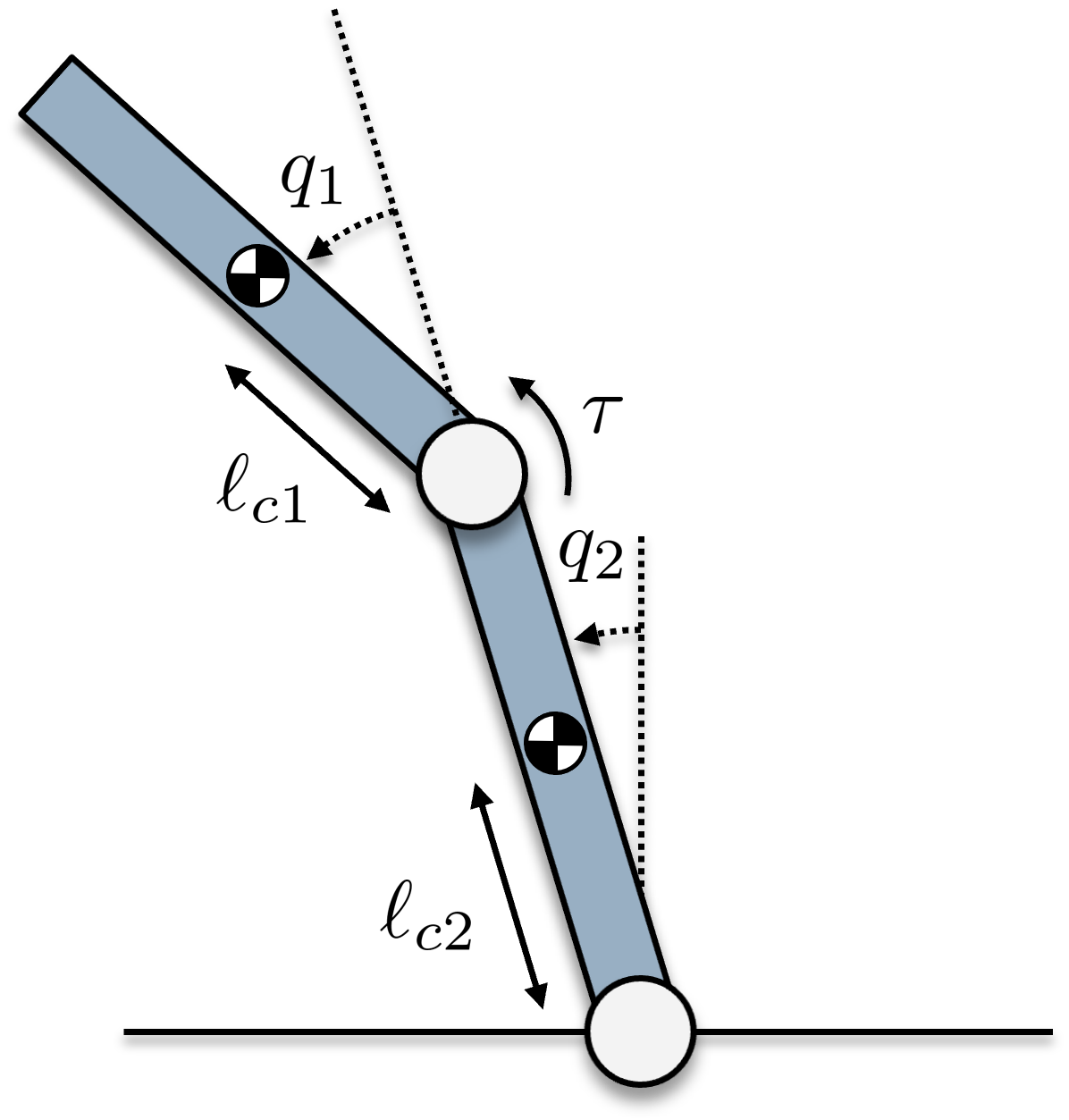}
	\caption{The acrobot system attempts to balance in the vertical position by manipulating the torque generated by an actuator between the two links.}
	\label{fig:acrobot}
\end{figure}
The acrobot system can be written as a pH system of the form \eqref{OLsystem} with
\begin{equation}\label{acrobot:system}
	\begin{split}
		M(q)
		&=
		\begin{bmatrix}
		c_2 & c_2 + c_3\cos q_1 \\
		c_2 + c_3\cos q_1 & c_1 + c_2 + 2c_3\cos q_1
		\end{bmatrix} \\
		V(q)
		&=
		c_4g\cos q_2 + c_5g\cos(q_1+q_2) \\
		G
		&=
		\begin{bmatrix}
			1 \\ 0
		\end{bmatrix},
	\end{split}
\end{equation}
where
\begin{equation}
	\begin{split}
		c_1 &= m_2\ell_{c2}^2 + m_1\ell_2^2 + J_{\ell 2} \\
		c_2 &= m_1\ell_{c1}^2 + J_{\ell 1} \\
		c_3 &= m_1\ell_2\ell_{c1} \\
		c_4 &= m_2\ell_{c2} + m_1\ell_{1} \\
		c_5 &= m_1\ell_{c1}.
	\end{split}
\end{equation}
For the purposes of simulation, we take the values $g = 9.8$, $c_1 = 2.3333, c_2 = 5.3333, c_3 = 2, c_4 = 3, c_5 = 2$ which were previously used in \cite{Mahindrakar2010}, \cite{Donaire2017a}.

In this example, the total energy-shaping controller proposed in \cite{Mahindrakar2010} is reconstructed as a CbI control scheme by solving the matching conditions of Proposition \ref{prop:matchingEquations}. In that work, the closed-loop mass matrix was chosen to be the constant matrix
\begin{equation}\label{MdPrevious}
	\begin{split}
		M_d^{-1}
		&=
		\begin{bmatrix}
			0.3385 & -0.9997 \\
			-0.9997 & 5.9058 \\
		\end{bmatrix}
	\end{split}
\end{equation}
which will be recovered in subsequent computations. 

The mass matrix of the acrobot system depends only on $q_1$, the actuated coordinate. The added inverse mass matrix is assumed to be a function of only $q_1$ also, resulting in the structure
\begin{equation}
	M_a^{-1}(q_1)
	=
	\begin{bmatrix}
		m_{a11}(q_1) & m_{a21}^\top(q_1) \\
		m_{a21}(q_1) & m_{a22}(q_1)
	\end{bmatrix}.
\end{equation}
As the system is underactuated degree 1 and the mass matrix is a function of only one variable, the kinetic energy matching equations can be reduced to an ODE as per Corollary \ref{corr:KE_ODE}. The resulting ODE has the form \eqref{underactuated1ODE} with $q_i = q_1$ and where $m_{a11}(q_1)$ is a free function. 

In order to recover the result \eqref{MdPrevious}, this free function $m_{a11}(q_1)$ is chosen as 
\begin{equation}
	\begin{split}
		m_{a11}(q_1)
		&=
		G^\top\left[M_d^{-1} - M^{-1}(q_1)\right]G \\
		&=
		0.3385 - \frac{c_1 + c_2 + 2c_3\cos q_1}{c_1c_2 - c_3^2\cos^2(q_1)}.
	\end{split}
\end{equation}
The initial conditions $m_{a12}(0), m_{a22}(0)$ are similarly defined as
\begin{equation}
	\begin{split}
		m_{a12}(0)
		&=
		G^\perp\left[M_d^{-1} - M^{-1}(0)\right]G = -0.1313 \\
		m_{a12}(0)
		&=
		G^\perp\left[M_d^{-1} - M^{-1}(0)\right]G^{\perp\top} = 5.2743.
	\end{split}
\end{equation}
The added inverse mass was evaluated numerically and the results are shown in Figure \ref{fig:acro_addedMass}. As the previously reported solution \eqref{MdPrevious} is globally defined, it is unsurprising that the inverse added mas is also globally defined. As expected, the minimum eigenvalue of $M^{-1} + M_a^{-1}$ is constant also.
\begin{figure}[h!]
	\centering
	\includegraphics[width=1.0\columnwidth]{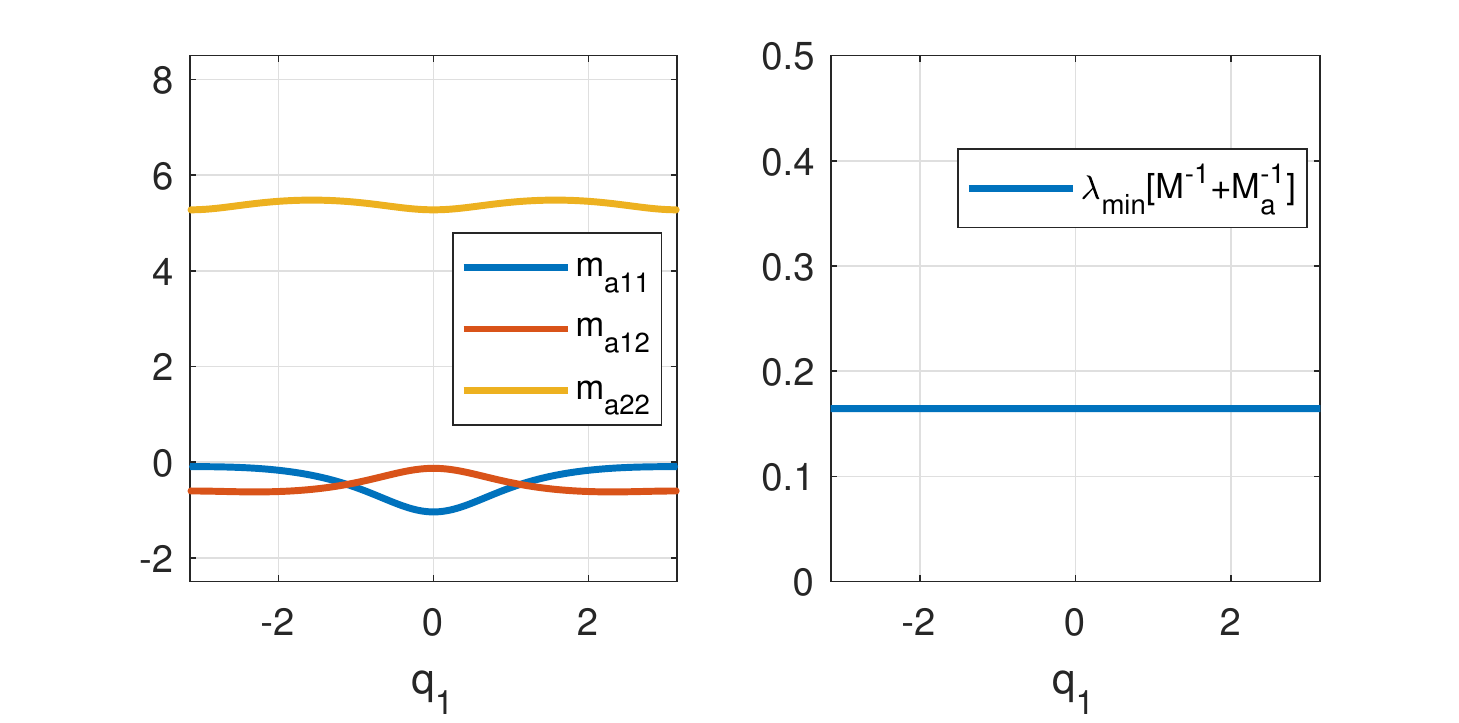}
	\caption{A solution for the added mass $M_a^{-1}$ for the acrobot was found to exist globally.}
	\label{fig:acro_addedMass}
\end{figure}

Solving the potential energy PDE \eqref{matching:PE_aug} is difficult due to the open-loop potential energy being a function of both $q_1$ and $q_2$. This dependence implies that $V_m$ cannot be resolved directly using an ODE solver. Considering the structure of $V$ in \eqref{acrobot:system}, it is proposed that the closed-loop energy $V_m$ has the structure
\begin{equation}\label{acrobot:VmDef}
		V_m(q)
		=
		f_1(q_1)\sin(q_2) + f_2(q_1)\cos(q_2),
\end{equation}
which has derivatives
\begin{equation}\label{acrobot:VmDeriv}
	\begin{split}
		\nabla_{q_1}V_a
		=&
		\frac{\partial f_1}{\partial q_1}\sin(q_2) + \frac{\partial f_2}{\partial q_1}\cos(q_2) \\
		\nabla_{q_2}V_a
		=&
		f_1(q_1)\cos(q_2) - f_2(q_1)\sin(q_2).
	\end{split}
\end{equation}
The open-loop potential energy has gradients
\begin{equation}\label{acrobot:Vderiv}
	\begin{split}
		\nabla_{q_1}V
		=&
		-c_5g\sin(q_1)\cos(q_2) - c_5g\cos(q_1)\sin(q_2) \\
		\nabla_{q_2}V
		=&
		-c_4g\sin(q_2) - c_5g\sin(q_1)\cos(q_2) \\
		&- c_5g\cos(q_1)\sin(q_2).
	\end{split}
\end{equation}
Substituting the expressions \eqref{acrobot:VmDeriv} and \eqref{acrobot:Vderiv} into \eqref{matching:PE} and matching coefficients results in the system of equations
\begin{equation}\label{acrobot:fODE}
	\begin{split}
		\begin{bmatrix}
			\frac{\partial f_1}{\partial q_1} \\ \frac{\partial f_2}{\partial q_1}
		\end{bmatrix}
		&=
		\frac{1}{s_2(q_1)} \\
		&\times
		\begin{bmatrix}
			c_4 g s_1(q_1) + c_5 g s_1(q_1)\cos(q_1) + s_3(q_1)f_2(q_1) \\
			c_5 g s_1(q_1)\sin(q_1) - s_3(q_1)f_1(q_1)
		\end{bmatrix},
	\end{split}
\end{equation}
which can be evaluated numerically. The values of $f_1, f_2$ at the origin should be chosen to ensure that the origin is an equilibrium point and $V_m$ is positive in $q_2$. Considering the expressions \eqref{acrobot:Vderiv}, \eqref{acrobot:fODE}, the origin is an equilibrium for $f_1(0) = 0$. The energy function \eqref{acrobot:VmDef} is locally positive with respect to $q_1$ for $f_2(0)$ negative. For the purpose of simulation, $f_2(0) = -50$ was used. The resulting function $V_m$ is shown in Figure \ref{fig:acro_addedPotEnergyVm}. 
\begin{figure}[h!]
	\centering
	\includegraphics[width=1.0\columnwidth]{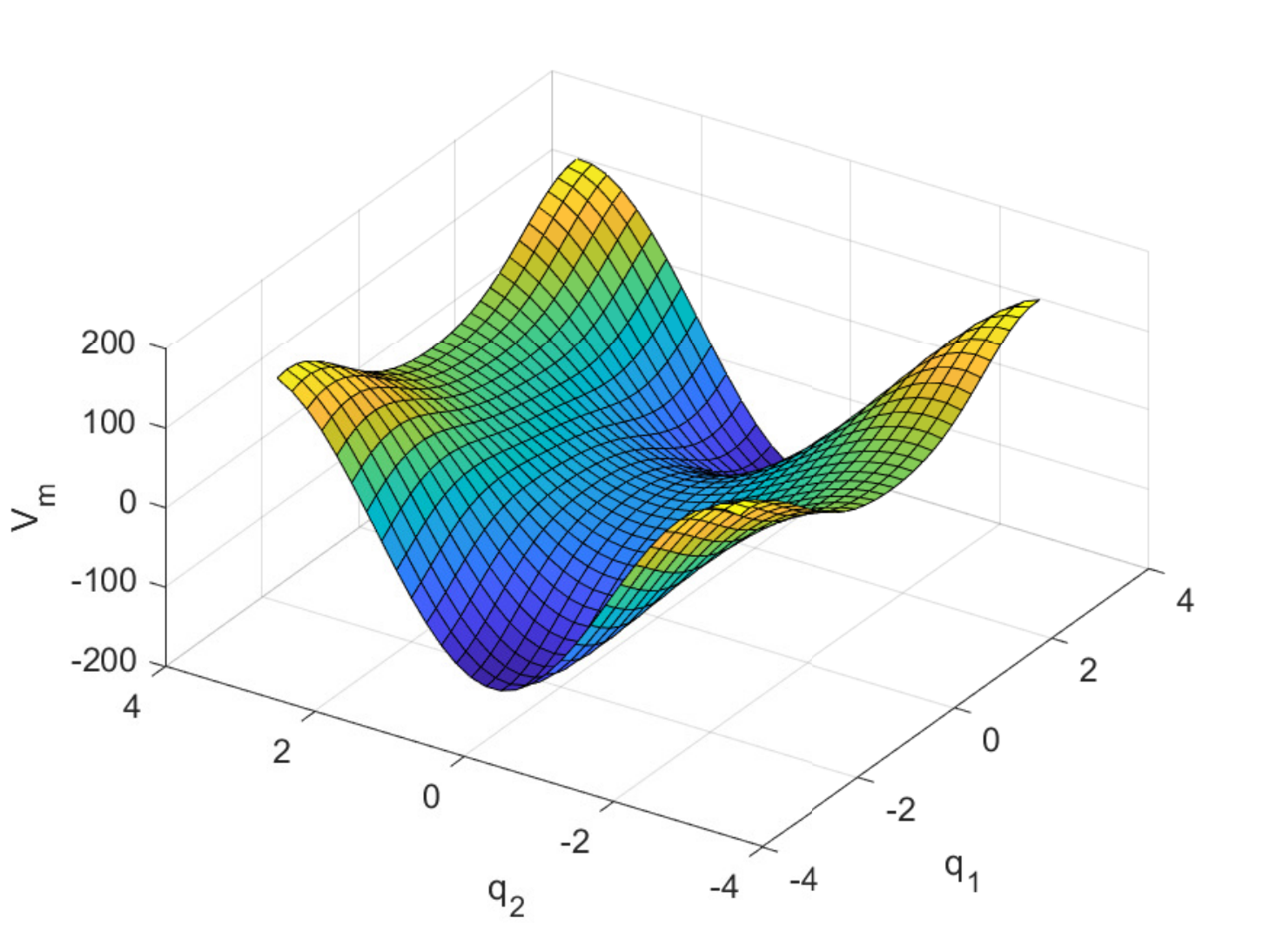}
	\caption{Solution for the added potential energy $V_m(q_2)$.}
	\label{fig:acro_addedPotEnergyVm}
\end{figure}

Considering Figure \ref{fig:acro_addedPotEnergyVm}, it is clear that $V_m$ is not positive definite with respect to the origin. Note, however, that $q_2 = 0$ has been stabilised. To ensure stability of $q_1 = 0$ also, the free term $V_f(\Gamma(q))$, defined in \eqref{VdDecomp}, is constructed. The function $\Gamma(\cdot)$ defined by the integral \eqref{GammaDef}, where $K(q)$ is a free function chosen to ensure solvability. Noting that $M^{-1}, M_a^{-1}$ are functions of $q_1$ only, the parametrisation
\begin{equation}
	\begin{split}
		\begin{bmatrix}
			\beta_1(q_1) & \beta_2(q_1)
		\end{bmatrix}
		=
		G^\top\left[M_a^{-1}(q_1) + M^{-1}(q_1)\right]M(q_1)
	\end{split}
\end{equation}
is introduced. The free function is chosen as $K(q_1) = \frac{1}{\beta_2(q_1)}$, resulting in 
\begin{equation}
	\begin{split}
		\Gamma(q)
		&=
		\int
		\begin{bmatrix}
			\frac{\beta_1(q_1)}{\beta_2(q_1)} & 1
		\end{bmatrix}
		\ dq \\
		&=
		\int \frac{\beta_1(q_1)}{\beta_2(q_1)} \ dq_1 + q_2,
	\end{split}
\end{equation}
which can be solved numerically from the initial condition $\Gamma(0_{2\times 1}) = 0$. The function $V_f(\cdot)$ was taken as $V_f(\Gamma(q)) = \frac12\kappa\Gamma(q)^2$ with $\kappa = 250$ for simulation. A contour plot of the resulting closed-loop potential energy on a log scale is shown in Figure \ref{fig:acro_addedPotEnergyVd}. Note that a minimum has been assigned to $q = 0_{2\times 1}$. As a final control design stage, damping is injected via the new passive input/output pair with
\begin{equation}
	v
	=
	-5G^\top (M_a^{-1} + M^{-1})p.
\end{equation}
The complete control signal is defined by the expression \eqref{controlDef}.
\begin{figure}[h!]
	\centering
	\includegraphics[width=1.0\columnwidth]{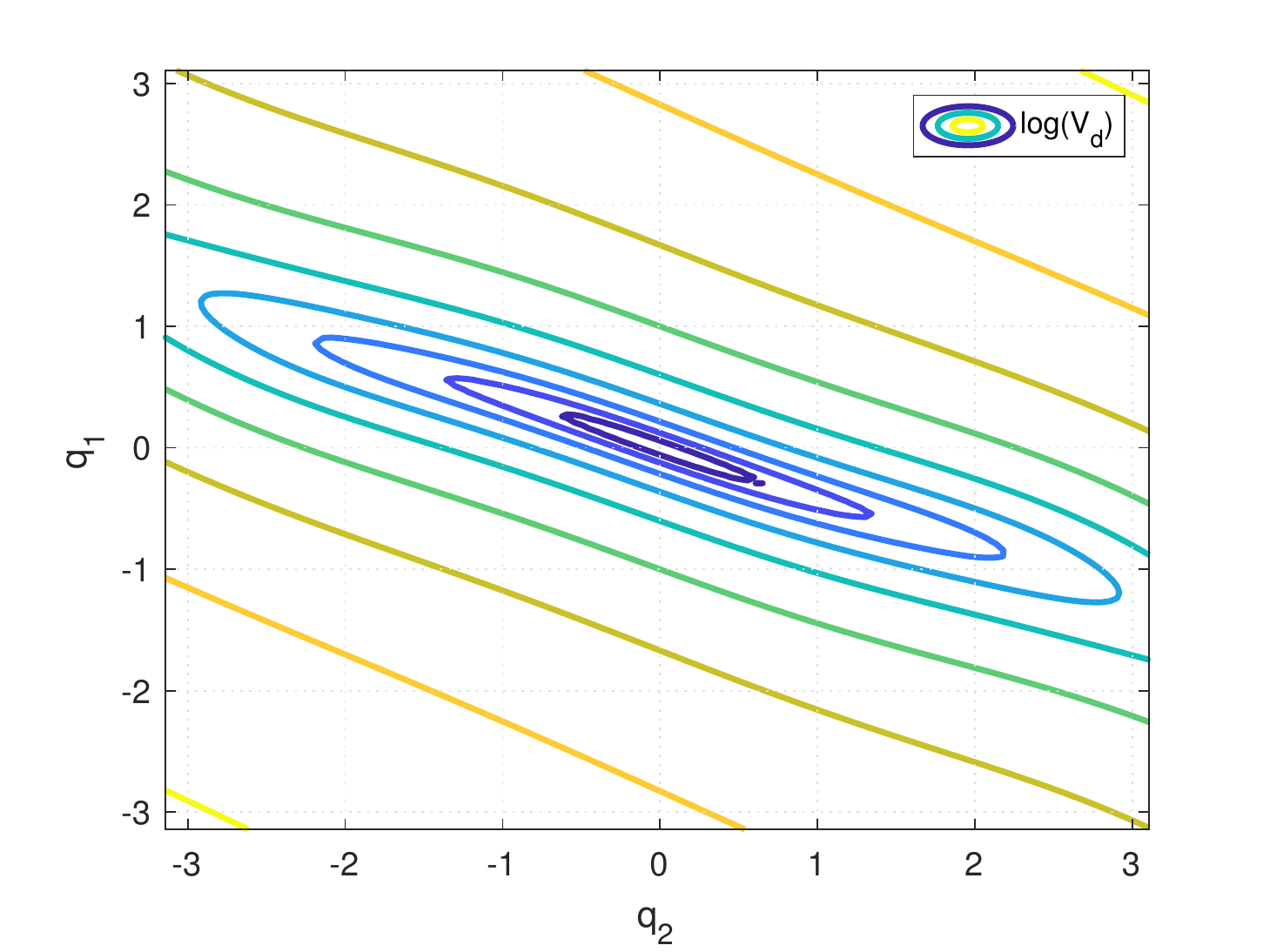}
	\caption{Contour plot of the closed-loop potential energy $V_d(q) = V_m(q) + V_f(\Gamma(q))$ on log scale for the acrobot system.}
	\label{fig:acro_addedPotEnergyVd}
\end{figure}

The acrobot system was simulated for 20 seconds from initial conditions $q(0) = (0, 0.5)$, $p(0) = (0,0)$. The resulting state evolution and closed-loop energy $H_d$ is shown in Figure \ref{fig:acro_simulation}. As expected, the proposed controller stabilises the origin and the closed-loop energy $H_d$ decreases monotonically.
\begin{figure}[h!]
	\centering
	\includegraphics[width=1.0\columnwidth]{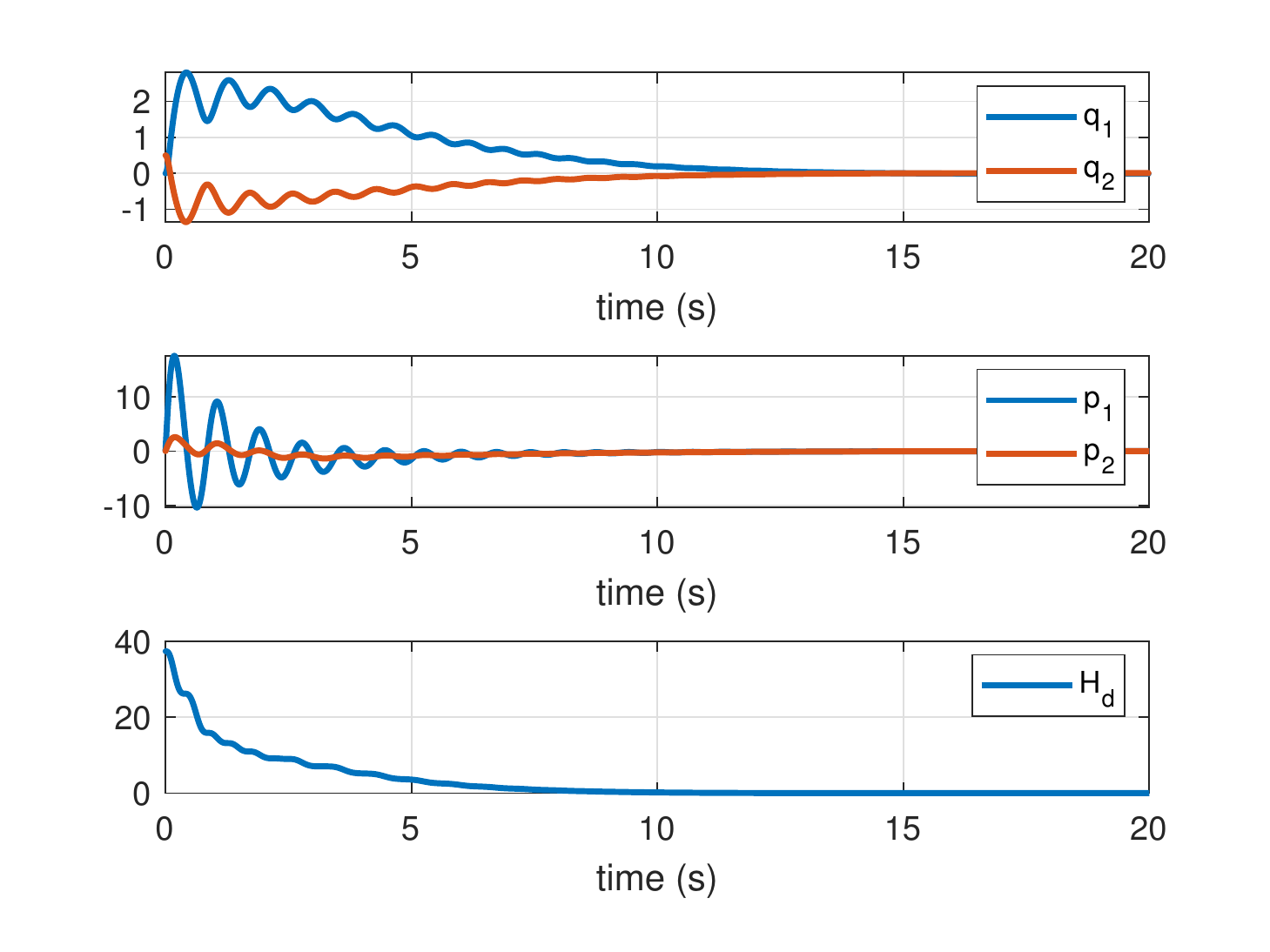}
	\caption{Numerical simulation of acrobot system in closed-loop with CbI scheme.}
	\label{fig:acro_simulation}
\end{figure}

\section{CONCLUSIONS AND FUTURE WORKS}\label{sec:Conclusion}
In this work total energy shaping has been shown to have a CbI interpretation which results in alternate matching equations related to the added inverse mass. These equations were utilised to construct controllers for the cart-pole and acrobot systems, both of which have the property that the mass matrix depends on only one variable, using numerical methods. While the proposed approach is effective, a number of technical aspects of this approach require further investigation. In particular:
\begin{itemize}
\item As detailed in Corollary \ref{corr:KE_ODE}, The kinetic energy matching equations can be posed as ODEs in the special case that the mass matrix depends on only one configuration variable. This property allows the matching equations to be evaluated numerical using ODE solvers. Further investigation into solving the matching equations in the case that the mass matrix is a function of multiple configuration variables is required. In some cases it may be possible to decouple the dependence on each coordinate, recovering equivalent ODEs. Alternatively, the numerical evaluation of the matching PDEs should be investigated.
\item When evaluating the kinetic energy matching equations in \eqref{underactuated1ODE}, the term $m_{a11}(q_i)$ is a free function that can be used to control the resulting added inverse mass. As seen in the cart-pole example of Section \ref{Sec:example:cartPole}, poor choice of this function results in the solution only being defined on a small domain. Conversely, in the acrobot example of Section \ref{Sec:example:cartPole} this term was chosen to ensure a global solution to the matching equations. Choice of this function defines a nonlinear control problem that should be investigated to ensure desirable behaviour of the result.
\item In both examples of Section \ref{sec:examples} the controllers were designed to stabilise the origin of the respective systems. While this was achieved and verified numerically, asymptotic stability was not established. Asymptotic stability requires that the passive output of the closed-loop system is zero-state detectable, a task that is non-trivial for underactuated systems. Further investigation into methods for injecting damping into the unactuated momentum channels of the closed-loop system is required. It is hoped that the CbI interpretation of the controller shown in Figure \ref{underactuatedMassShaping} may provide new insight into how this might be achieved.
\end{itemize}








\bibliographystyle{IEEEtran}

\end{document}